\documentclass[journal]{IEEEtran}
\usepackage{verbatim}
\usepackage{CJK}
\usepackage{indentfirst}
\usepackage{amsmath}
\usepackage{amsthm}
\usepackage{graphicx}
\usepackage{subfigure}
\usepackage{tabularx}
\usepackage{multirow}
\usepackage{algorithm}
\usepackage{algpseudocode}
\usepackage{xcolor}

\begin{document}
\title{An Evolutionary Game With the Game Transitions Based on the Markov Process}

\author{Minyu Feng, ~\IEEEmembership{Member,~IEEE,}
        Bin Pi,
        Liang-Jian Deng, ~\IEEEmembership{Senior Member,~IEEE,}
        and J\"{u}rgen Kurths  % <-this % stops a space
\thanks{This work was supported in part by the National Nature Science Foundation of China (NSFC) under Grant No. 62206230 and No. 12271083, in part by the Humanities and Social Science Fund of Ministry of Education of the People's Republic of China under Grant No. 21YJCZH028, and in part by the Natural Science Foundation of Sichuan Province under Grant No. 2022NSFSC0501.}

\thanks{Minyu Feng is with the College of Artificial Intelligence, Southwest University, Chongqing 400715, China.

Bin Pi and Liang-Jian Deng are with the School of Mathematical Sciences, University of Electronic Science and Technology of China, Chengdu 611731, China (e-mail: liangjian.deng@uestc.edu.cn).

J\"{u}rgen Kurths is with the Potsdam Institute for Climate Impact Research, 14437 Potsdam, Germany, and also with the Institute of Physics, Humboldt University of Berlin, 12489 Berlin, Germany.}}% <-this % stops a space
\maketitle \thispagestyle{plain}\pagestyle{plain}

\begin{abstract}
The psychology of the individual is continuously changing in nature, which has a significant influence on the evolutionary dynamics of populations. To study the influence of the continuously changing psychology of individuals on the behavior of populations, in this paper, we consider the game transitions of individuals in evolutionary processes to capture the changing psychology of individuals in reality, where the game that individuals will play shifts as time progresses and is related to the transition rates between different games. Besides, the individual's reputation is taken into account and utilized to choose a suitable neighbor for the strategy updating of the individual. Within this model, we investigate the statistical number of individuals staying in different game states and the expected number fits well with our theoretical results. Furthermore, we explore the impact of transition rates between different game states, payoff parameters, the reputation mechanism, and different time scales of strategy updates on cooperative behavior, and our findings demonstrate that both the transition rates and reputation mechanism have a remarkable influence on the evolution of cooperation. Additionally, we examine the relationship between network size and cooperation frequency, providing valuable insights into the robustness of the model.
\end{abstract}

\begin{IEEEkeywords}
Network evolutionary game, Game transitions, Markov process, Reputation
\end{IEEEkeywords}
\maketitle \thispagestyle{plain}\pagestyle{plain}

% \fi
%\IEEEpeerreviewmaketitle

\thispagestyle{empty}

%\vspace{-0.5\baselineskip}
\section{Introduction}

\IEEEPARstart{C}{ooperation} is often understood as the prosocial behavior of bearing a cost to provide a benefit to another individual \cite{01}, which implies that selfish and unrelated individuals should improve the payoff of others at their own cost \cite{02}. However, according to Darwin's theory of evolution \cite{03}, defection should be the result of natural selection in a social dilemma, but there exist many cooperative behaviors in human society and nature. For example, the worker bees of honey bees will sacrifice their lives to protect the colony. Vampire bats will allow their companions to suck their own blood to help preserve the lives of their companions when necessary. Therefore, understanding the emergence and maintenance of cooperation in a competitive world is significant in explaining some of the prosocial behavior in nature, which has attracted the attention of researchers from a wide range of fields \cite{04}.

With the emergence of classical networks such as the small-world network proposed by Watts and Strogatz in 1998 \cite{05} and the scale-free network proposed by Barab$\acute{a}$si and Albert in 1999 \cite{06}, various novel network models have been proposed \cite{07}. For instance, a novel evolving network model has been established considering the growing and decreasing process based on the queueing system \cite{08}. Li et al. \cite{09} introduced an evolving population network through the migration of people and studied the spread of disease on networks. Three novel models based on the homogeneous Poisson, nonhomogenous Poisson, and birth death process were proposed to reveal the influence of the vertex generating mechanism of complex networks \cite{10}. The novel complex networks can be a theoretical tool for our study of the social behavior of structured populations.

At the same time, numerous game models have been proposed to describe different types of social dilemmas in order to characterize the game process between individuals in a population, among which the famous are the snowdrift game (SDG) \cite{11}, \cite{12}, prisoner's dilemma game (PDG) \cite{13}, \cite{14}, and stag-hunt game (SHG) \cite{15}, \cite{16}. We can apply game theory to understand social behavior from a new perspective. The network evolutionary game theory, which combines game theory with complex networks, has become one of the most useful frameworks for studying the emergence of cooperation in populations. Moreover, a large number of mechanisms that can facilitate the evolution of cooperative behavior have also been investigated, among which the five rules proposed by Nowak in 2006 \cite{17} are well-known. In addition, other mechanisms such as the behavior of conformity \cite{18}, \cite{19}, memory mechanisms \cite{20}, \cite{21}, rewards and punishments \cite{22}, \cite{23}, etc, have also been shown to play a significant role in the maintenance of cooperation.

However, it is worth noting that in most previous studies, it has been assumed that the games played by individuals are deterministic, which can be interpreted as the game is constant all the time. In fact, determinism is only a special case, while stochasticity is a common phenomenon in life. For example, the payoffs between enterprises do not remain the same all the time although they adopt the same strategy but instead change dynamically over time. In other words, the game between individuals should not just stay in one state as time progresses, but it should also evolve over time. In the past few years, some researchers have studied this issue. For example, Su et al. \cite{24} introduced game transitions into classical models of evolutionary dynamics and found that game transitions can significantly reduce the critical benefit-to-cost threshold for cooperation to evolve in social dilemmas. Hilbe et al. \cite{25} utilized the theory of stochastic games and evolutionary game theory to analyze the proposed idea and obtained that the dependence of the public resource on previous interactions can significantly increase the propensity to cooperate. Moreover, as is well-known, reputation is an efficient and ubiquitous social control mechanism in natural societies, which is also crucial in the interaction of individuals. In general, individuals who are willing to help others will have a higher reputation, and they will also be more likely to receive help from others. Some researchers \cite{26} have taken reputation into account in dynamic network modeling, where individuals decide whether to break an edge based on reputation and found that the evolutionary game on this can effectively lead to the formation of cooperator clusters. Hu et al. \cite{27} performed Monte Carlo simulations on social networks to determine critical values of the degree of rationality and the reputation threshold that warrants high levels of trust and social wealth. Luo et al. \cite{28} simulated the evolution of the environmental governance cooperative behavior of enterprises considering the supervision behavior of the government and the reputation evaluation behavior of environmental social organizations.

In this paper, we consider that the game state of each individual is not fixed, but will change over time. We note that unlike the previous one on game transitions, in our novel model, they are not caused by changes in the environment, but each individual will change from one game state to another at a rate that follows an exponential distribution, thus characterizing the changing psychology of the individual. We utilize the exponential distribution since there are extensive studies \cite{29}, \cite{30} showing that there are human behaviors that obey exponential distributions. For example, Liang et al. \cite{31} built models for 20 million trajectories with fine granularity collected from more than 10 thousand taxis in Beijing and found that the taxis' traveling displacements in urban areas tend to follow an exponential distribution instead of a power-law. In other words, the game state of an individual can be thought of as a Markov chain in continuous time, and its state space is $\{G_0, G_1, ..., G_n\}$, where $G_i$ ($i=0, 1, \cdots, n$) denotes a particular game model. Hereby, we emphasize two confusing concepts: evolutionary game and game evolution. In this paper, the evolutionary game focuses on the updating of individuals' strategies, while game evolution focuses on the updating of individuals' game states, which illustrates that they are different concepts. In addition, we also take the reputation of the individual into account when individuals update their strategies. Specifically, each individual $j$ has a reputation value $Re_j$, which is updated according to the individual's previous strategy: if the individual cooperates in the last round, $Re_j$ will increase in a step of $\delta$, and vice versa, while $Re_j$ also has an upper or lower limit, i.e., reputation cannot go up or down indefinitely. Moreover, in order to fit better with reality, individuals prefer to interact with individuals with high reputations, i.e., they are more likely to choose individuals with high reputations for payoff comparison and strategy learning. Generally, the main contributions of this paper can be summarized as follows:
\begin{itemize}
\item We employ the game transitions of individuals based on the Markov process to describe the changing psychology of individuals in reality, where the game that individuals will play shifts as time passes and is related to the corresponding transition rates between different games.
\item We introduce the reputation mechanism to allow individuals to choose highly reputable neighbors with a high probability for strategy learning, and it evolves dynamically according to the individual's strategy.
\item The statistical number of individuals staying in different game states is studied and the relative error between the simulation results and theoretical results is calculated, it is found to be small, which indicates the correctness of the theory we proposed.
\item The impacts of transition rates between different game states and payoff parameters with and without reputation mechanism on the evolution of cooperation are investigated, and the results demonstrate that both transition rates and reputation mechanism facilitate the maintenance of cooperation.
\item The robustness of the model is verified by investigating the impact of different network sizes on the cooperation ratio, and the influence of different processes of strategy updating time on the cooperative behavior is also studied.
\end{itemize}

The remainder of this paper is structured as follows: Firstly, we present the evolutionary dynamics with game transitions of individuals and provide related theoretical analysis in Sec. \ref{part II}. Then, we carry out simulations to illustrate the validity of our model and explore the evolution of cooperation under different conditions in Sec. \ref{part III}. Finally, we summarize our results, and further give outlooks in Sec. \ref{part IV}.

%\vspace{-0.5\baselineskip}
\section{Evolutionary game on complex networks based on Markov method}  \label{part II}

In this section, we propose a novel evolutionary game model on complex networks regarding the game transitions as well as the reputations of individuals. Concretely, considering that individuals in reality do not always play one game without change, but have a time-varying mechanism, which can be understood as the individual's psychology changing all the time. In order to characterize this behavior, we introduce the Markov approach to represent the game transitions of individuals, which can be interpreted as each individual changing from one game to another with a certain probability. In reality, individuals in a population have various contact relationships with each other, and such relations can be described by complex networks, where a vertex in the network represents an individual and an edge indicates the interaction between two individuals \cite{32}, \cite{33}. Besides, the reputation mechanism of individuals is also considered, in which individuals are more inclined to choose individuals with high reputations for the payoff comparison to learn their strategies when conducting strategy updates.

Considering the reality that the game models of individuals will transform from one game model to another, we regard the game model utilized by an individual as a Markov chain, where the game model represents the state space of the Markov chain. In other words, the future game model of an individual is independent of the past game models and only relies on the current game model, but it is also related to the past game models while the correlation is so weak that it is almost negligible. Besides, this assumption has been proved to be reasonable by many studies \cite{34}, \cite{35}. In the following analysis, we present the Markov chain of individual game states, introduce the reputation into the game between network individuals, and give some definitions and theoretical analysis to obtain the limit distribution of the $n$ game states.

%\vspace{-0.5\baselineskip}
\subsection{Markov Chain of Individual Game States}

First, we introduce the process of individual game state transitions, which is depicted by a Markov chain. The game state of each individual $j$ in the population can be considered as a Markov chain, denoted as $\{X_j(t), t\geq 0\}$ with the state space $E_j=\{G_0, G_1, \cdots, G_n\}$, where the integers $0, 1, \cdots, n$ describe $n$ different game models, respectively, such as the prisoner's dilemma game, snowdrift game, stag hunt game, etc. It is worth mentioning that the individual's game model is the individual's game state, i.e., they are equivalent. Besides, we suppose that an individual will transform the game state $G_i$ into the game state $G_{i+1}$ at an exponential rate, $\lambda_i$, whereas the game state $G_{i+1}$ transitions to $G_i$ at an exponential transition rate $\mu_{i+1}$. This means that the duration of an individual holding a game state in the network follows an exponential distribution, and it will become another game state if the duration finishes. As an example, Fig. \ref{game state transition} clearly shows the transition of an individual game state, which can also be described as the game evolution of individuals. Suppose that all individuals in the network adopt the game state $G_0$ at the beginning of the evolutionary process, the individuals stay in the game state $G_0$ for a time duration following an exponential distribution with rate $\lambda_0$ and then turn the game state $G_1$. Subsequently, the individual will perform $G_2$ at an exponential rate $\lambda_1$ or $G_0$ at an exponential rate $\mu_1$.

\begin{center}
%\vspace{-1\baselineskip}
\begin{figure}[ht]
\centering
\includegraphics[scale=0.30]{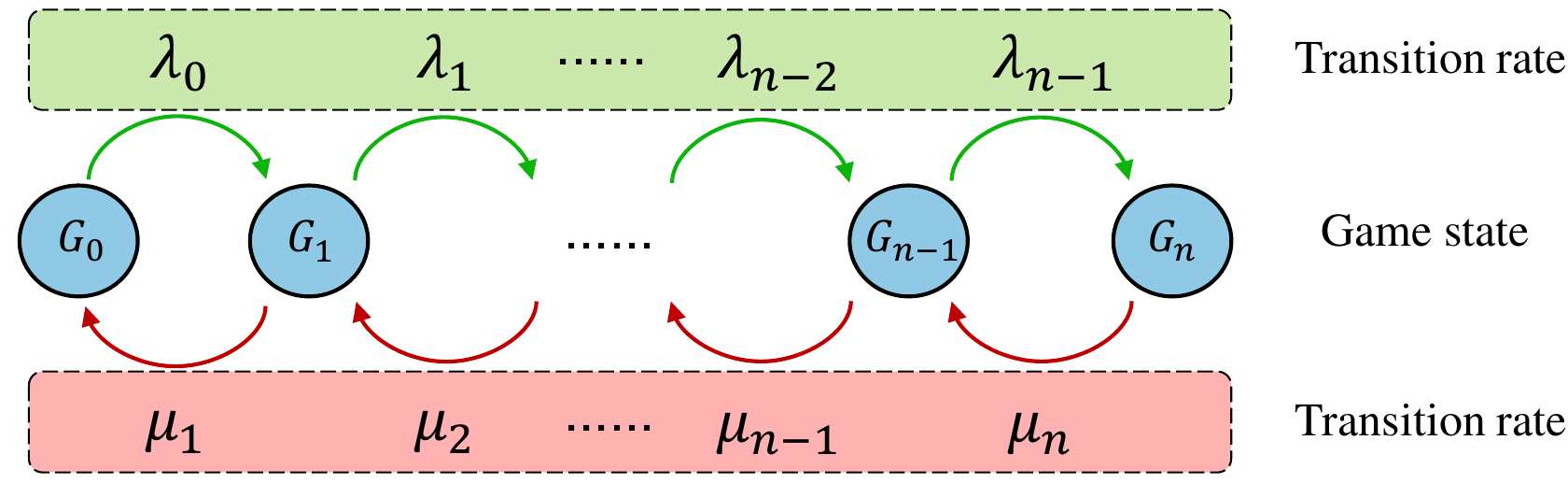}
\caption{\textbf{Game state transitions of individuals.} Each individual in the network will have the game state transition, the arrows signify the transitions from one game state to another, where the letters $\lambda_i$ and $\mu_{i+1}$ above or below the arrows denote the rate of transition from game state $G_i$ to $G_{i+1}$ and the rate of transition from game state $G_{i+1}$ to $G_i$ ($i=0, 1, \cdots, n-1$), respectively.}
\label{game state transition}
%\vspace{-2.5\baselineskip}
\end{figure}
\end{center}

%\vspace{-0.5\baselineskip}
\subsection{The Game Model and Strategy Evolution}

In this subsection, we will explain the interaction between individuals, including the calculation of individuals' payoffs and the evolution of strategies. In each evolutionary time step, each individual will play a certain game with each neighbor to obtain a payoff, which is related to the individual's game state and different strategic interactions. Specifically, when individuals are staying at game state $G_i$, mutual cooperation (C, C) leads to a ``reward'' of $R_i$ for the individual, while mutual defection (D, D) brings a ``punishment'' outcome to an individual payoff of $P_i$; unilateral cooperation results in a ``sucker's payoff" of $S_i$ for the cooperator, while for the defector, an individual receives a ``temptation to defect'' of $T_i$. Thus, the payoff matrix for an individual playing the game $G_i$ can be expressed as below:

\begin{small}
%\vspace{-0.5\baselineskip}
\begin{equation}
\label{Payoff matrix}
M_i = \left(
  \begin{array}{cc}
    R_i & S_i \\
    T_i & P_i  \\
  \end{array}
\right).
%\vspace{-0.5\baselineskip}
\end{equation}
\end{small}

We note that different orders of parameters ($R_i$, $S_i$, $T_i$, and $P_i$) will yield different social dilemmas as long as the four conditions ($R_i>P_i$, $R_i>S_i$, $2R_i>T_i+S_i$, and $T_i>R_i$ or $P_i>S_i$) are satisfied, which clearly portray the conflict between individual and collective benefits in social dilemmas. In particular, this conflict is manifested as follows: from a collective perspective, a combination of co-cooperation (C, C) strategies is better than unilateral cooperation (C, D) or unilateral defection (D, C); from an individual perspective, co-cooperation (C, C) is better than co-defection (D, D). However, under the effect of the individual's greed characteristic ($T_i > R_i$) or the fear of the game opponent's choice of defective strategy ($P_i > S_i$), the individual chooses the defective strategy, thus producing a result that conflicts with the collective benefits.

In order to better coincide with reality, we also take the reputation characteristics of the individuals into account. At the beginning of the evolution, each individual $j$ will be given a random number to represent his/her own reputation value $Re_j$. In addition, the individuals' reputations will evolve over time as well as their strategies and are related to the individuals' behaviors. The reputation possessed by the individual is evaluated by the third-party assessment system, which is similar to the credit agency in reality. Therefore, the reputation information about the individual is open to members of the public group, i.e., all individuals in the network can be aware of the reputations of their neighbors. Moreover, the reputations of individuals at time step $t + 1$ depend on their behavior at the time step $t$. The reputation of individual $j$ will increase by $\delta$ if the individual is a cooperator in the previous step, otherwise, it will decrease by $\delta$ if the individual is a defector in the previous step, in which $\delta$ indicates the unit of reputation change. Consequently, the update rule of the reputation of individual $j$ can be denoted as follows:

\begin{small}
%\vspace{-0.5\baselineskip}
\begin{equation}
\label{repuation update}
Re_{j}^{t+1}=\begin{cases}
	Re_{j}^{t}+\delta, s_{j}^{t}=C\\
	Re_{j}^{t}-\delta, s_{j}^{t}=D\\
\end{cases},
%\vspace{-0.5\baselineskip}
\end{equation}
\end{small}
where $s_j^t$ represents the strategy adopted by individual $j$ at the time step $t$.

For the update rule of individual's strategy, consistent with most previous studies, we utilize the Fermi process, i.e., individuals will adopt the strategy of individuals with higher payoffs than themselves with a higher probability. However, different from previous studies, when the individuals select another individual for payoff comparison and update their strategies, they do not randomly choose from their neighbors, but choose the individual with a higher reputation with a higher probability, i.e., the probability of an individual $j$ to choose another individual $j\prime$ is determined by $\Pi_{j\prime}$,

\begin{small}
%\vspace{-0.5\baselineskip}
\begin{equation}
\label{probability of choosing individual}
\Pi_{j\prime}=\frac{Re_{j\prime}}{\sum_{r\in \varGamma_j}{Re_r}},
%\vspace{-0.5\baselineskip}
\end{equation}
\end{small}
where $\varGamma_j$ is the set of neighbors of individual $j$. For the update of strategy, if individual $j$ chooses neighbor $j\prime$, then $j$ will adopt $j\prime$'s strategy in the next time step with the following probability,

\begin{small}
%\vspace{-0.5\baselineskip}
\begin{equation}
\label{probability of updating strategy}
P(s_{j} \leftarrow s_{j\prime})=\frac{1}{1+e^{(U_{j}-U_{j\prime}) / \kappa}},
%\vspace{-0.5\baselineskip}
\end{equation}
\end{small}
where $s_j$ and $U_j$ indicate the strategy and payoff of individual $j$, respectively, and $\kappa$ represents the noise factor, which is employed to portray the irrational choices of individuals in the game.

%\vspace{-0.5\baselineskip}
\subsection{Definitions and Theoretical Analysis}

In this subsection, we hereby give some required notations, definitions, and theoretical analysis for the expectation of the number of individuals playing a certain game state in the network when the evolution is stable. As stated above, the game state of each individual will change as time progresses, which will result in a change in the number of individuals playing a certain game in the network. Therefore, we give the definitions as follows:

\newtheorem{Def}{Definition}

\begin{Def}
\label{N(t)}
$\{N_i(t), t\geq0\}$ denotes a stochastic process of the number of individuals staying in the game state $G_i$ with the state space $\Omega_i=\{0, 1, \cdots, N\}$, where $N$ represents the scale of the network.
\end{Def}

In Def. \ref{N(t)}, the state space $\Omega_i$ describes the possible values of individuals in the network that stay in game state $G_i$ at each moment, where it takes the largest value of $N$, which means that all individuals in the network are playing game $G_i$ at this time, while the number of individuals in other game states equals to 0. Next, we define the probability of transferring the game state $G_x$ to $G_y$ of individual $j$ as below:

\begin{Def}
\label{transition probability}
$p_{x, y}^j(\bigtriangleup h)$ denotes the probability that the game state conducted by the individual $j$ is $G_x$ and will turn to $G_y$ in the time interval $\bigtriangleup h$, which can be represented as a conditional probability,

\begin{small}
%\vspace{-0.5\baselineskip}
\begin{equation}
\label{conditional probability}
p_{x, y}^{j}(\bigtriangleup h)=P\{X_j(t+\bigtriangleup h)=G_y|X_j(t)=G_x\},
%\vspace{-0.5\baselineskip}
\end{equation}
\end{small}
where $G_x$ and $G_y \in$ $\{G_0, G_1, \cdots, G_n\}$ as the game state space and $X_j(t)$ represents the game state conducted by the individual $j$ at time $t$.
\end{Def}

According to Def. \ref{transition probability}, we give the following definition of the rate that the game state conducted by the individual $j$ is $G_x$ and will next transfer to $G_y$:

\begin{Def}
\label{transition rate definition}
$q_{x, y}^j$ denotes the transition rate of different game states corresponding to the transition probability $p_{x, y}^j(\bigtriangleup h)$, which can be expressed as follows:

\begin{small}
%\vspace{-0.5\baselineskip}
\begin{equation}
\label{transition rate equation}
q_{x,y}^{j}=\begin{cases}
	\underset{t\rightarrow 0^+}{\lim}\frac{p_{x,y}^{j}\left( t \right)}{t}, x\ne y\\
	\underset{t\rightarrow 0^+}{\lim}\frac{1-p_{x,y}^{j}\left( t \right)}{t}, x=y\\
\end{cases},
%\vspace{-0.5\baselineskip}
\end{equation}
\end{small}
where $G_x$ and $G_y \in$ $\{G_0, G_1, \cdots, G_n\}$ as the game state space.
\end{Def}

Defs. \ref{transition probability} and \ref{transition rate definition} describe the transition probability and transition rate, respectively, for the game state of the individual $j$ in the network turning from $G_x$ to $G_y$. Then, we present the limit probability of the individual $j$ conducting the game state $G_y$ as follows:

\begin{Def}
\label{limiting probability definition}
The probability $p_y^j(t)$ that the individual $j$ conducts the game state $G_y$ at time $t$ is denoted as below:

\begin{small}
%\vspace{-0.5\baselineskip}
\begin{equation}
\label{probability distribution}
p_y^j(t)=P\{X_j(t)=G_y\},
%\vspace{-0.5\baselineskip}
\end{equation}
\end{small}
where $G_y \in$ $\{G_0, G_1, \cdots, G_n\}$ as the game state space. Besides, we have $p_y^j$ denote the limit probability as follows:

\begin{small}
%\vspace{-0.5\baselineskip}
\begin{equation}
\label{limiting probability equation}
p_{y}^{j}=\pi_{y}^{j}=\underset{t\rightarrow \infty}{\lim}p_{y}^{j}(t),
%\vspace{-0.5\baselineskip}
\end{equation}
\end{small}
if the probability converges to some values when $t\rightarrow \infty$.
\end{Def}

In Def. \ref{limiting probability definition}, we depict the steady-state probability or the stationary probability distribution of the game state conducted by individual $j$ in the network, which means that we can get the probability that the game state of individual $j$ is $G_y$ when $t\rightarrow \infty$.

Notably, the input process for the game states $G_0$ and $G_n$ is a single Poisson process with variable rates since the input rate is only related to the number of individuals currently in $G_1$ and $G_{n-1}$, respectively. However, the input for game state $G_i$ ( $1\leq i\leq n-1$) is a component of the input of individuals from the game states $G_{i-1}$ and $G_{i+1}$, which indicates that it is a compound Poisson flow.

In order to prove the stationarity of the number of individuals performing a certain game state in the network, we first carry out two lemmas. In the game state space $\{G_0, G_1, \cdots, G_n\}$, we have the following lemmas to hold:

\newtheorem{lem}{Lemma}
\begin{lem}
\label{Lemma 1}
Supposing that the limit probability $p_{y}^{j}$ exists, then we have
\begin{small}
%\vspace{-0.5\baselineskip}
\begin{equation}
\underset{t\rightarrow \infty}{\lim}p_{y}^{j}(t)^{\prime}=0.
%\vspace{-0.5\baselineskip}
\end{equation}
\end{small}
\end{lem}
\begin{proof}
Assume that there exists game state $G_x \in \{G_0, G_1, \cdots, G_n\}$, which enables $\underset{t\rightarrow \infty}{\lim}p_{x}^{j}(t)^{\prime}=s>0$ to satisfy. Therefore, through the definition of limitation, there exists a $t_1$, for arbitrary $\epsilon$ and those $t\geq t_1$ that makes $\left| p_{x}^{j}(t)^{\prime}-s \right| <\epsilon$ satisfy, which suggests that $p_{x}^{j}(t)^{\prime}>s-\epsilon$. Thus, we can derive the following limitation:
\begin{small}
%\vspace{-0.5\baselineskip}
\begin{equation}
\underset{t\rightarrow \infty}{\lim}p_{x}^{j}(t)>p_x^j\left( t_1 \right) +\underset{t\rightarrow \infty}{\lim}\left( s-\epsilon \right) \left( t-t_1 \right) =\infty,
%\vspace{-0.5\baselineskip}
\end{equation}
\end{small}
which contradicts the definition of probability that needs to belong to $[0, 1]$. Therefore, our previous assumption does not hold, i.e., $\underset{t\rightarrow \infty}{\lim}p_{y}^{j}(t)^{\prime}=0$.

The result follows.
\end{proof}

Lemma \ref{Lemma 1} shows that if we derivate the probability of an individual staying in the game state $G_y$ at $t\rightarrow \infty$, then the derivative is 0, and this lemma is employed in subsequent proofs of the existence of stationary distribution. According to Eq. \ref{conditional probability}, we know that the transition probability of the game state of an individual is independent of the starting time $t$ and is only related to the time interval $\bigtriangleup h$, which indicates that the Markov chain is homogeneous. Besides, for any individual $j$ in any state $G_x, G_y \in \{G_0, G_1, \cdots, G_n\}$, there is always a $t_1$ such that $p_{x,y}^{j}(t_1)>0$ follows; at the same time, there is also a certain $t_2$ following $p_{y,x}^{j}(t_2)>0$ satisfies, which means that all game states are communicated, i.e., the condition of irreducibility follows. Additionally, again according to Eq. \ref{conditional probability}, we have $\underset{t\rightarrow 0^+}{\lim}p_{y,y}^{j}\left( t \right) = 1$, while $\underset{t\rightarrow 0^+}{\lim}p_{x,y\ne x}^{j}\left( t \right) = 0$, namely
\begin{small}
%\vspace{-0.5\baselineskip}
\begin{equation}
\underset{t\rightarrow 0^+}{\lim}p_{x,y}^{j}\left( t \right) =\begin{cases}
	0,x\ne y\\
	1,x=y\\
\end{cases},
%\vspace{-0.5\baselineskip}
\end{equation}
\end{small}
which represents that the Markov chain is continuous. Briefly, the Markov chain $\{X_j(t), t\geq0\}$ of individual game state is homogeneous, irreducible, and continuous, which enables us to derive the existence of its limit distribution.

\begin{lem}
\label{Lemma 2}
For the homogeneous, irreducible, and continuous Markov chain $\{X_j(t), t\geq0\}$ with the state space $E_j$, its stationary distribution $\{\pi_y^j, G_y\in E_j\}$ exists and follows
\begin{small}
%\vspace{-0.5\baselineskip}
\begin{equation}
\label{stationarity equation}
-\pi _{y}^{j}q_{y}^{j}+\sum_{G_k\ne G_y\in E_j}{\pi_{k}^{j}q_{k,y}^{j}}=0,
%\vspace{-0.5\baselineskip}
\end{equation}
\end{small}
where $q_{y}^{j}$ indicates the transition rate of individual $j$ from the game state $G_y$ to other game states.
\end{lem}

Remarkably, we suppose that the order of limit and summation can be swapped in the proof of Lemma \ref{Lemma 2}, but this does not always satisfy. However, they are valid in most models, including the birth-death process and all finite state models. Moreover, based on Eq. \ref{stationarity equation}, we can equivalently translate it to $\pi _{y}^{j}q_{y}^{j}=\sum_{G_k\ne G_y\in E_j}{\pi _{k}^{j}q_{k,y}^{j}}$, where $\pi _{y}^{j}q_{y}^{j}$ can be understood as the rate at which the process leaves game state $G_y$ and $\sum_{G_k\ne G_y\in E_j}{\pi _{k}^{j}q_{k,y}^{j}}$ can be interpreted as the rate at which the process enters game state $G_y$. Therefore, in other words, Eq. \ref{stationarity equation} is also a statement that the rate at which a process enters and leaves game state $G_y$ is equal, from which we can get the stationary distribution of the system.

Subsequently, we perform the stationary distribution of the stochastic process $X_j(t)$ based on Lemma \ref{Lemma 2}.

\newtheorem{thm}{Theorem}

\begin{thm} \label{limiting thm}
In our proposed model, for the stochastic process $X_j(t)$ of individual game state, let $t\rightarrow \infty$, its stationary distribution exists, and follows
\begin{small}
%\vspace{-0.5\baselineskip}
\begin{equation}
\pi _{0}^{j}=\frac{1}{1+\sum_{r=1}^n{\frac{\lambda _0\lambda _1\cdots \lambda _{r-1}}{\mu _1\mu _2\cdots \mu _r}}}
%\vspace{-0.5\baselineskip}
\end{equation}
\end{small}
and
\begin{small}
%\vspace{-0.5\baselineskip}
\begin{equation}
\pi _{k}^{j}=\frac{\lambda _0\lambda _1\cdots \lambda _{k-1}}{\mu _1\mu _2\cdots \mu _k}\pi _{0}^{j}, 1\le k\le n.
%\vspace{-0.5\baselineskip}
\end{equation}
\end{small}
\end{thm}

According to Thm. \ref{limiting thm}, we can calculate the probability that the game state of individual $j$ staying in is $G_k$ when $t\rightarrow \infty$, i.e., the stationary distribution of the stochastic process $X_j(t)$, and this theorem can further yield the expected number of individuals conducting each game state in the network.

\begin{thm}  \label{expected thm}

For the game state $G_k$, the expected number of individuals in the network staying in it is
\begin{small}
%\vspace{-0.5\baselineskip}
\begin{equation}
\label{expected number0}
E\left[ G_0 \right] =\frac{N}{1+\sum_{r=1}^n{\frac{\lambda _0\lambda _1\cdots \lambda _{r-1}}{\mu _1\mu _2\cdots \mu _r}}}
%\vspace{-0.5\baselineskip}
\end{equation}
\end{small}
and
\begin{small}
%\vspace{-0.5\baselineskip}
\begin{equation}
E\left[ G_k \right] =\frac{N\lambda _0\lambda _1\cdots \lambda _{k-1}}{\mu _1\mu _2\cdots \mu _k\left( 1+\sum_{r=1}^n{\frac{\lambda _0\lambda _1\cdots \lambda _{r-1}}{\mu _1\mu _2\cdots \mu _r}} \right)},1\le k\le n,
%\vspace{-0.5\baselineskip}
\end{equation}
\end{small}
where $N$ denotes the scale of the network.
\end{thm}

\begin{proof}
According to Thm. \ref{limiting thm}, we get the stationary distribution of the game state for each individual. And because the game state of each individual in the network is i.i.d. (independent and identically distributed), the expected number of individuals staying in game state $G_0$ can be obtained by multiplying the total number of individuals by the stationary distribution of individuals staying in game state $G_0$, i.e.,
\begin{small}
%\vspace{-0.5\baselineskip}
\begin{equation}
E\left[ G_0 \right] = N \times \pi _{0}^{j} = \frac{N}{1+\sum_{r=1}^n{\frac{\lambda _0\lambda _1\cdots \lambda _{r-1}}{\mu _1\mu _2\cdots \mu _r}}}.
%\vspace{-0.5\baselineskip}
\end{equation}
\end{small}

Analogously, the expected number of individuals staying in game state $G_k$ can be expressed as follows:

\begin{small}
%\vspace{-0.5\baselineskip}
\begin{equation}
\begin{aligned}
E\left[ G_k \right] &= N \times \pi _{k}^{j} = \frac{N\lambda _0\lambda _1\cdots \lambda _{k-1}}{\mu _1\mu _2\cdots \mu _k\left( 1+\sum_{r=1}^n{\frac{\lambda _0\lambda _1\cdots \lambda _{r-1}}{\mu _1\mu _2\cdots \mu _r}} \right)}, \\
&1\le k\le n.
\end{aligned}
%\vspace{-0.5\baselineskip}
\end{equation}
\end{small}

The results follow.

\end{proof}

In Thm. \ref{expected thm}, we obtain the expected number of individuals performing each game state in the network based on the stationary distribution of individuals staying in the game state in Thm. \ref{limiting thm}, which is only related to the network size and the rates in the game transitions.

\begin{center}
%\vspace{-1\baselineskip}
\begin{figure}[htbp]
\centering
\includegraphics[scale=0.25]{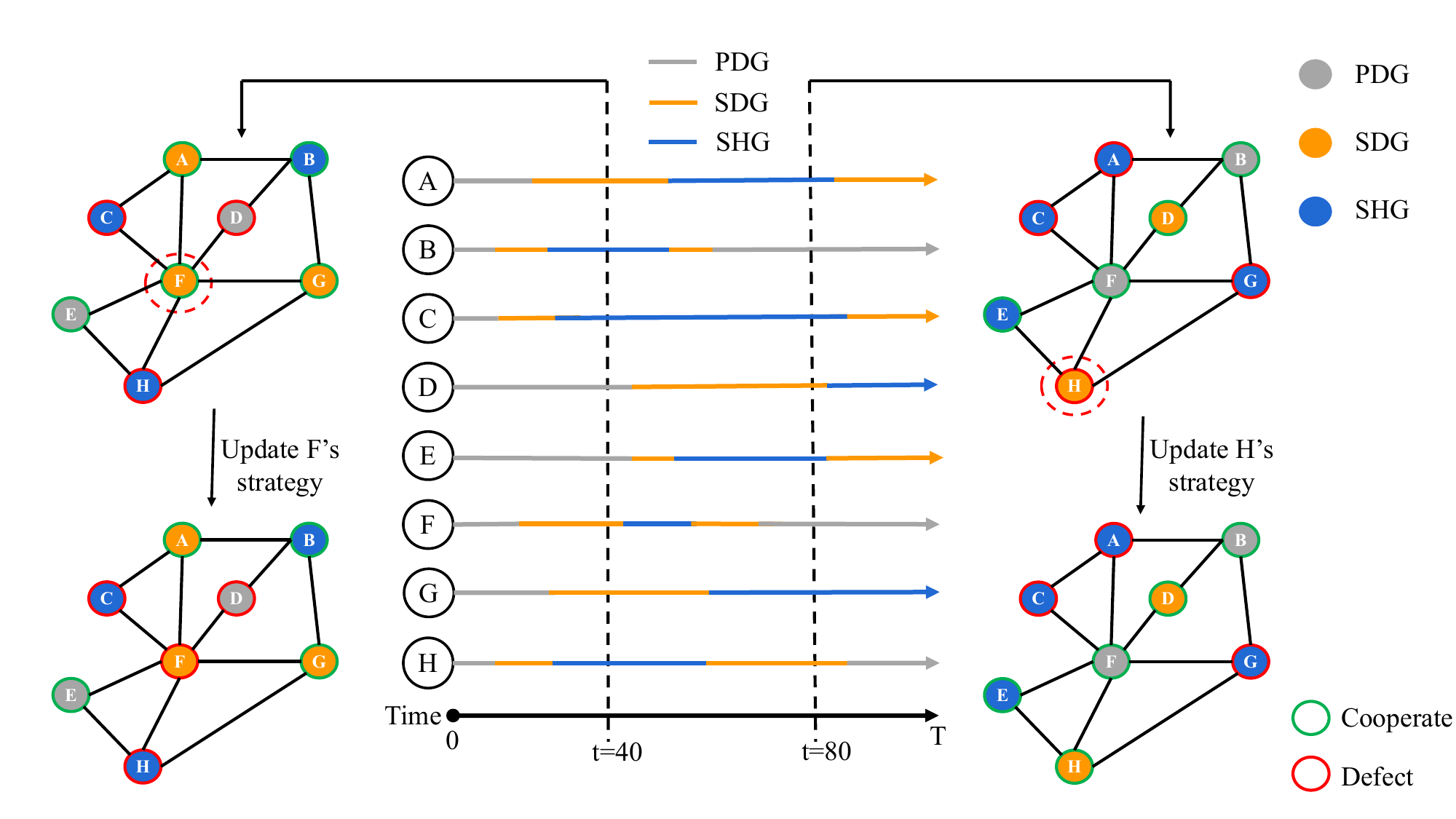}
\caption{\textbf{An illustration of the model.} Grey, orange, and blue durations represent the prisoner's dilemma game (PDG), the snowdrift game (SDG), and the stag-hunt game (SHG) durations of individuals, respectively. The red or green border indicates that the strategy adopted by the individual is defection or cooperation, respectively. The game state of each individual will be transformed from one to another at a specific rate. We chose evolution times $t = 40$ and $t = 80$ to observe snapshots of the network. The individual circled by the red dashed line updates his/her strategy at the next moment, and he/she will choose the more reputable individual among his/her neighbors with a higher probability to compare his/her payoff and decide whether to adopt the neighbor's strategy.}
\label{an example}
\end{figure}
\end{center}

\begin{center}
\begin{table*}[htbp]
\renewcommand{\arraystretch}{1.5}
\caption{Three different types of game states along with their payoff parameters and dilemma descriptions}
\begin{center}
\begin{tabular}{ccc}
\hline\noalign{\smallskip}
Game states & Parameters  & Dilemmas \\
\noalign{\smallskip}\hline\noalign{\smallskip}
Stag-hunt game & $R=1$\textgreater{}$T=r$\textgreater{}$P=0$\textgreater{}$S=-r$    & Players prefer mutual defection to unilateral cooperation. \\
Snowdrift game & $T=1+r$\textgreater{}$R=1$\textgreater{}$S=1-r$\textgreater{}$P=0$ & Players prefer unilateral defection to mutual cooperation.  \\
Weak prisoner's dilemma game & $T=b$\textgreater{}$R=1$\textgreater{}$P=S=0$      & Players prefer mutual defection to unilateral cooperation. \\
\noalign{\smallskip}\hline
\end{tabular}
\label{payoff parameter}
\end{center}
\end{table*}
\end{center}

Overall, in this section, we have described our model in detail, including the Markov chain of individual game states, the game model and strategy evolution, and some definitions and theoretical analysis. Besides, we give an example in Fig. \ref{an example} to make our model more explicit. There are eight individuals in the network, grey, orange, and blue represent individuals located in different game states, and the different colored borders of individuals indicate the different strategies taken by the individuals. We choose evolution times $t=40$ and $t=80$ to observe snapshots of the network and notice the strategy evolution of an individual, which is circled by the red dashed line. Specifically, both evolutionary game and game evolution are involved in this figure. A change in the border color indicates the updating of the strategy, i.e., a round of the evolutionary game, while a change in the internal color represents the updating of the game state, i.e., a round of game evolution. In Fig. \ref{an_update}, we show a time-axis representation of strategy updates and game transitions for individual F. The gray, orange, and blue colors on the time-axis indicate that the current game played by F are PDG, SDG, and SHG, respectively. The times $t_i (0 < i \leq 6)$ represent the moments when F updates strategy. It is important to note that an individual's game transition time is determined by a specific transition rate, while the timing of strategy updates can follow different rules. For instance, the time interval for strategy update can be fixed to a value or it can adhere to a certain distribution. As a result, the time scales of game transitions and strategy updates of individuals are different. In the following section, we will implement simulations to verify the validity of our model and analyze the evolution of cooperative behavior in the network.

\begin{center}
\begin{figure}[htbp]
\centering
\includegraphics[scale=0.25]{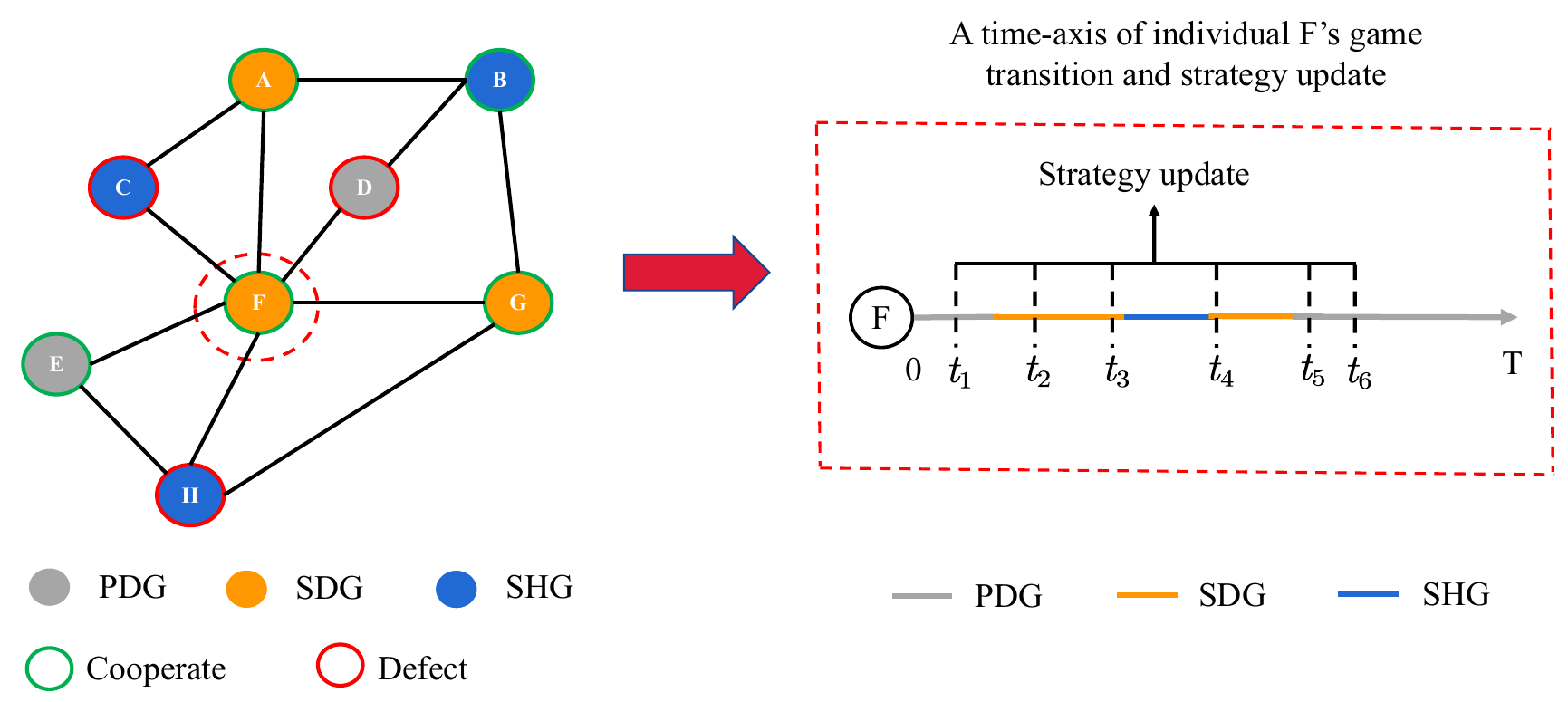}
\caption{\textbf{An example diagram of the evolutionary time of an individual's game state and strategy.} This figure illustrates the game transition and strategy update process for individual F. The game transition of individual F is determined by his/her specific transition rate, while the timing of strategy updates can follow different rules.}
\label{an_update}
\end{figure}
\end{center}
%\vspace{-3.5\baselineskip}

\begin{center}
\begin{figure*}[htbp]
\centering
\subfigure[initial transition rates]{
\includegraphics[scale=0.11]{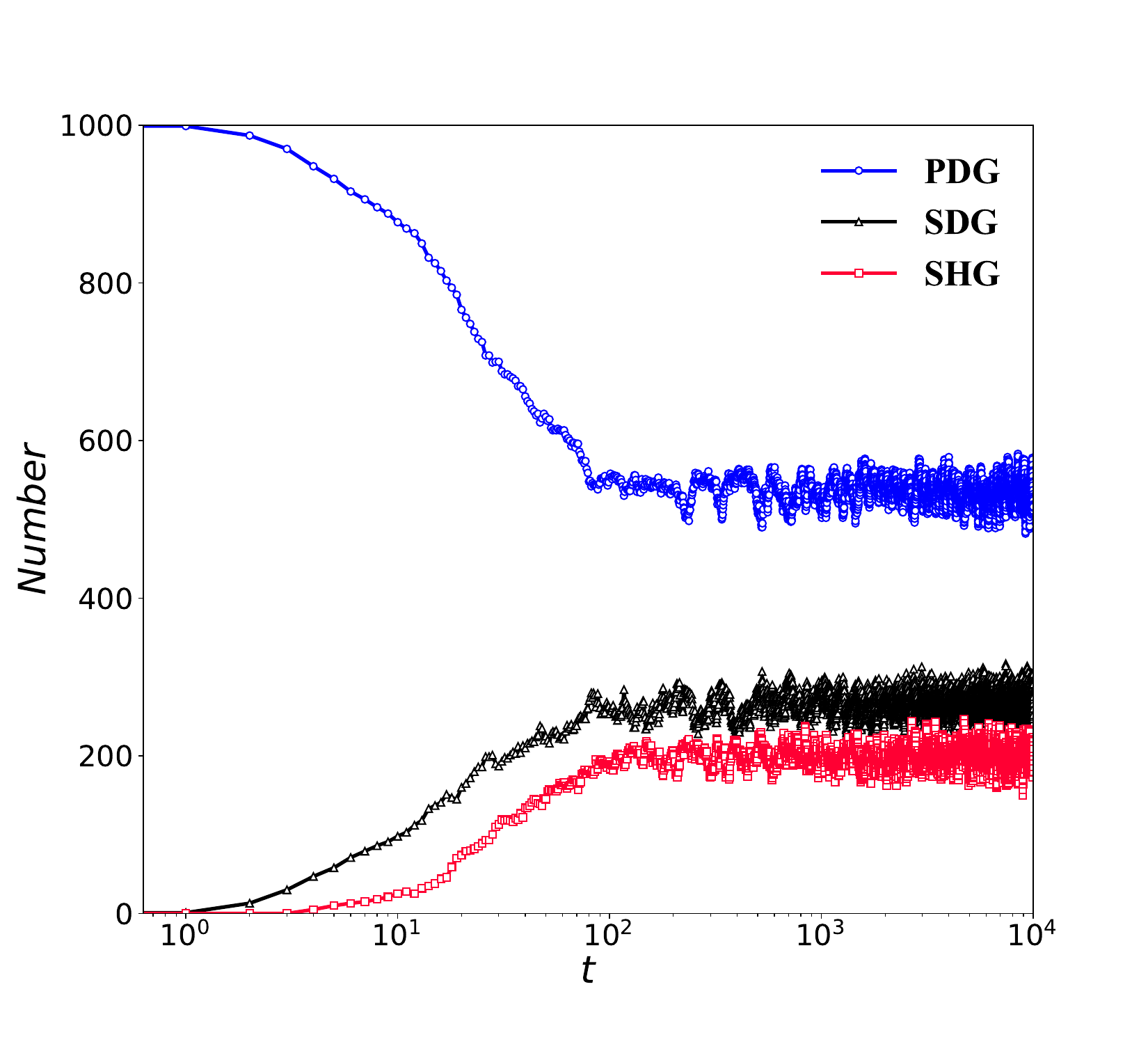}
\label{t_origin}}
\subfigure[$\lambda_0$=0.04]{
\includegraphics[scale=0.11]{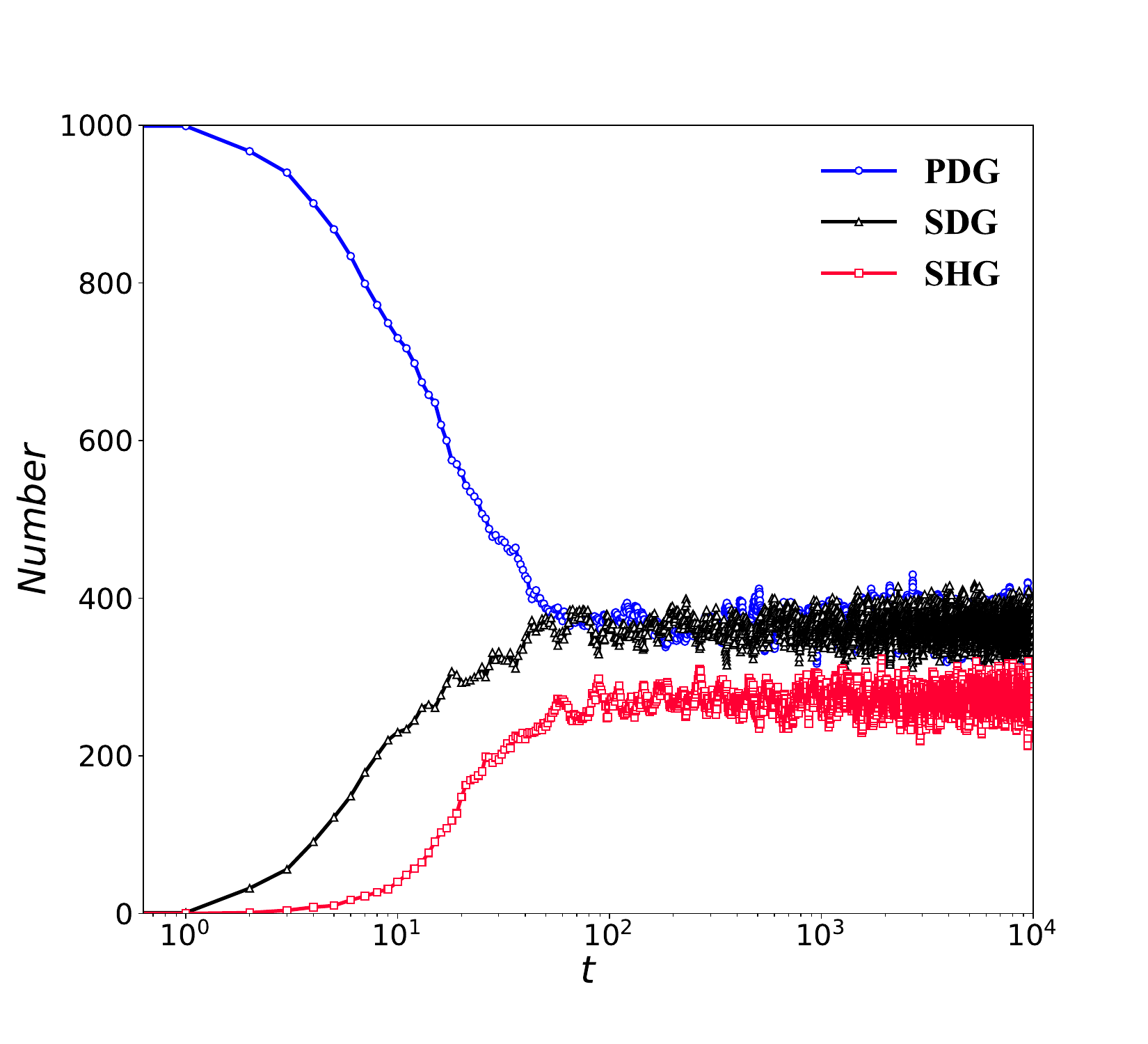}
\label{t_lam0}}
\subfigure[$\lambda_1$=0.12]{
\includegraphics[scale=0.11]{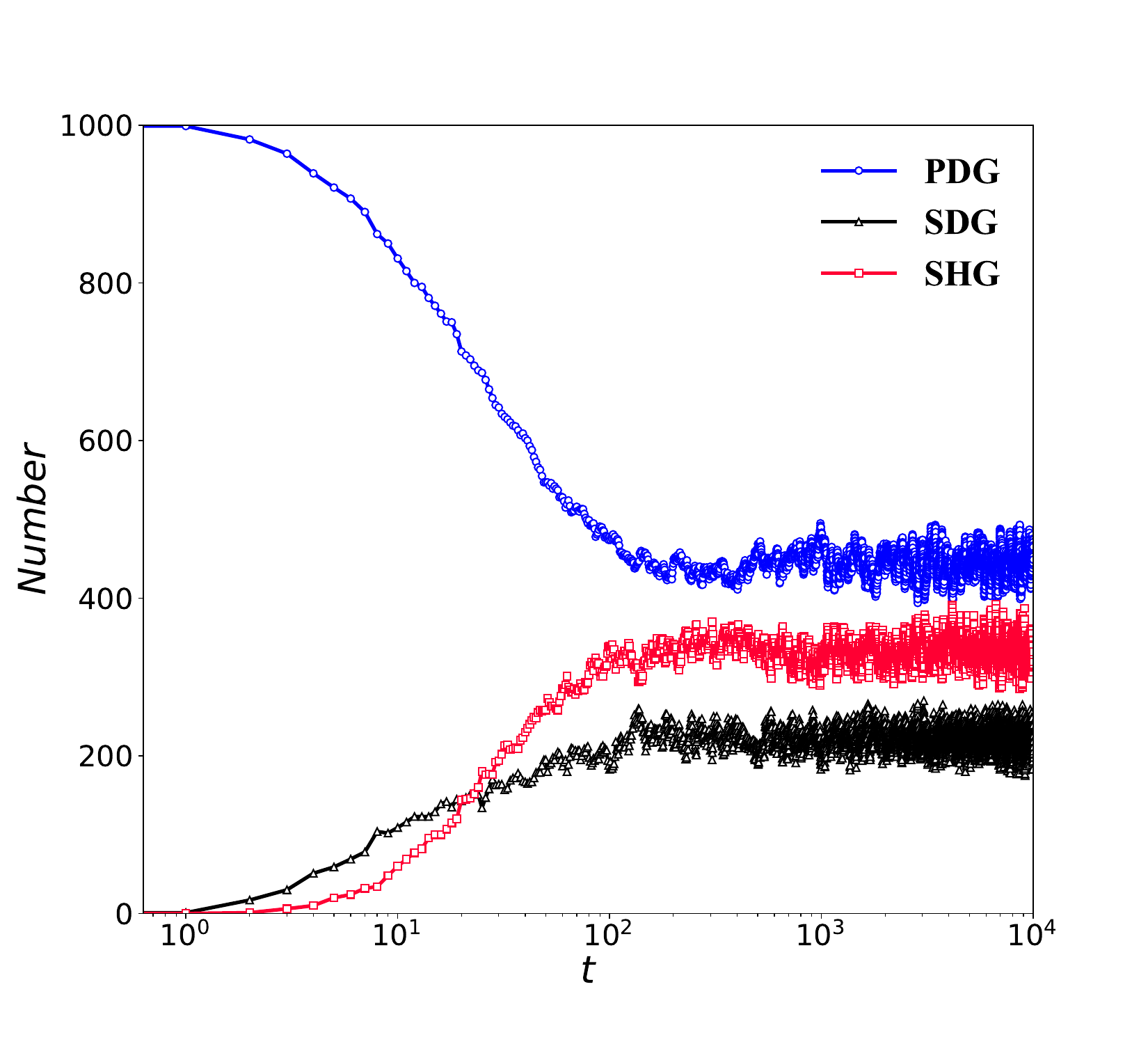}
\label{t_lam1}}
\subfigure[$\mu_1$=0.08]{
\includegraphics[scale=0.11]{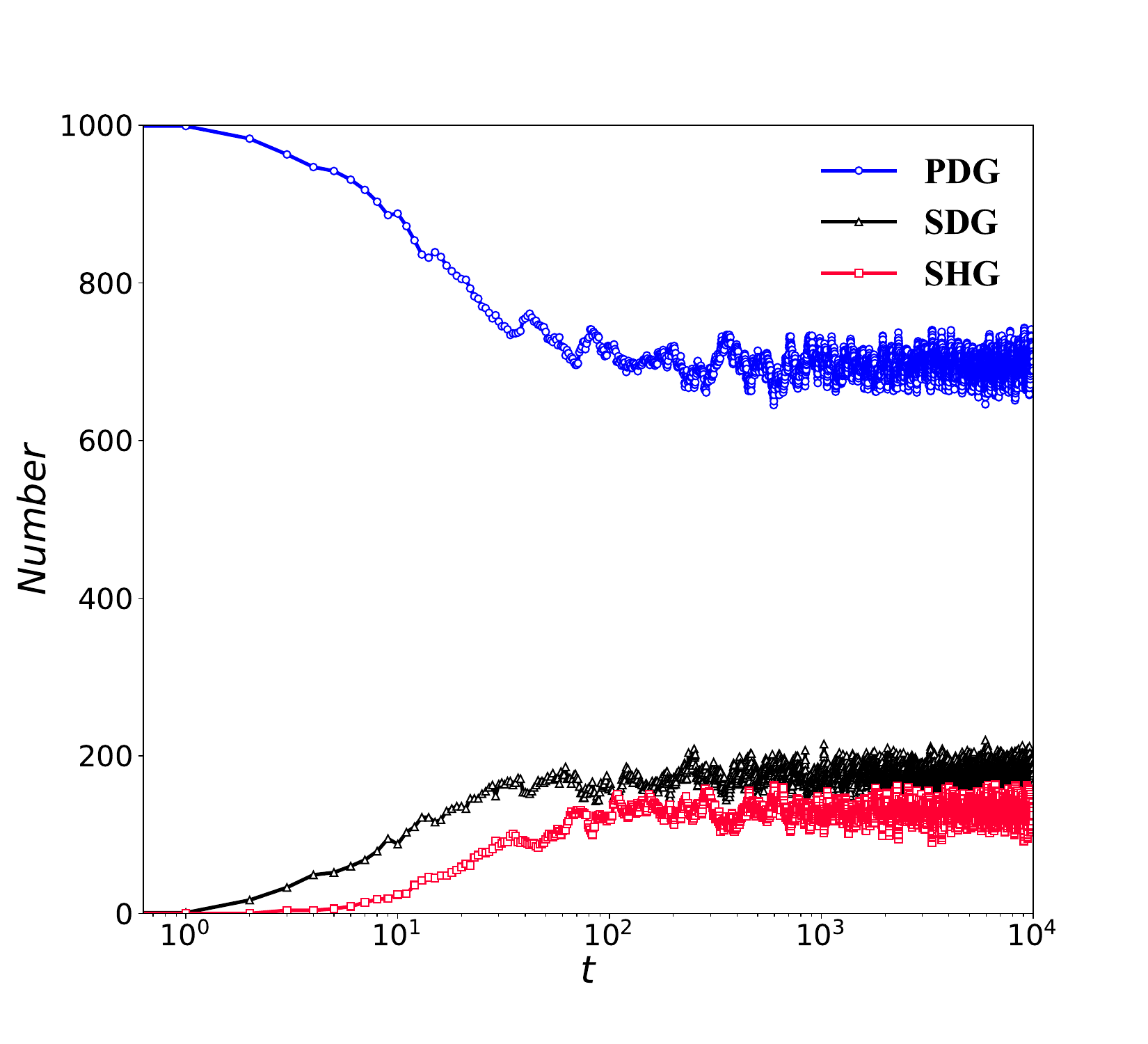}
\label{t_mu1}}
\subfigure[$\mu_2$=0.16]{
\includegraphics[scale=0.11]{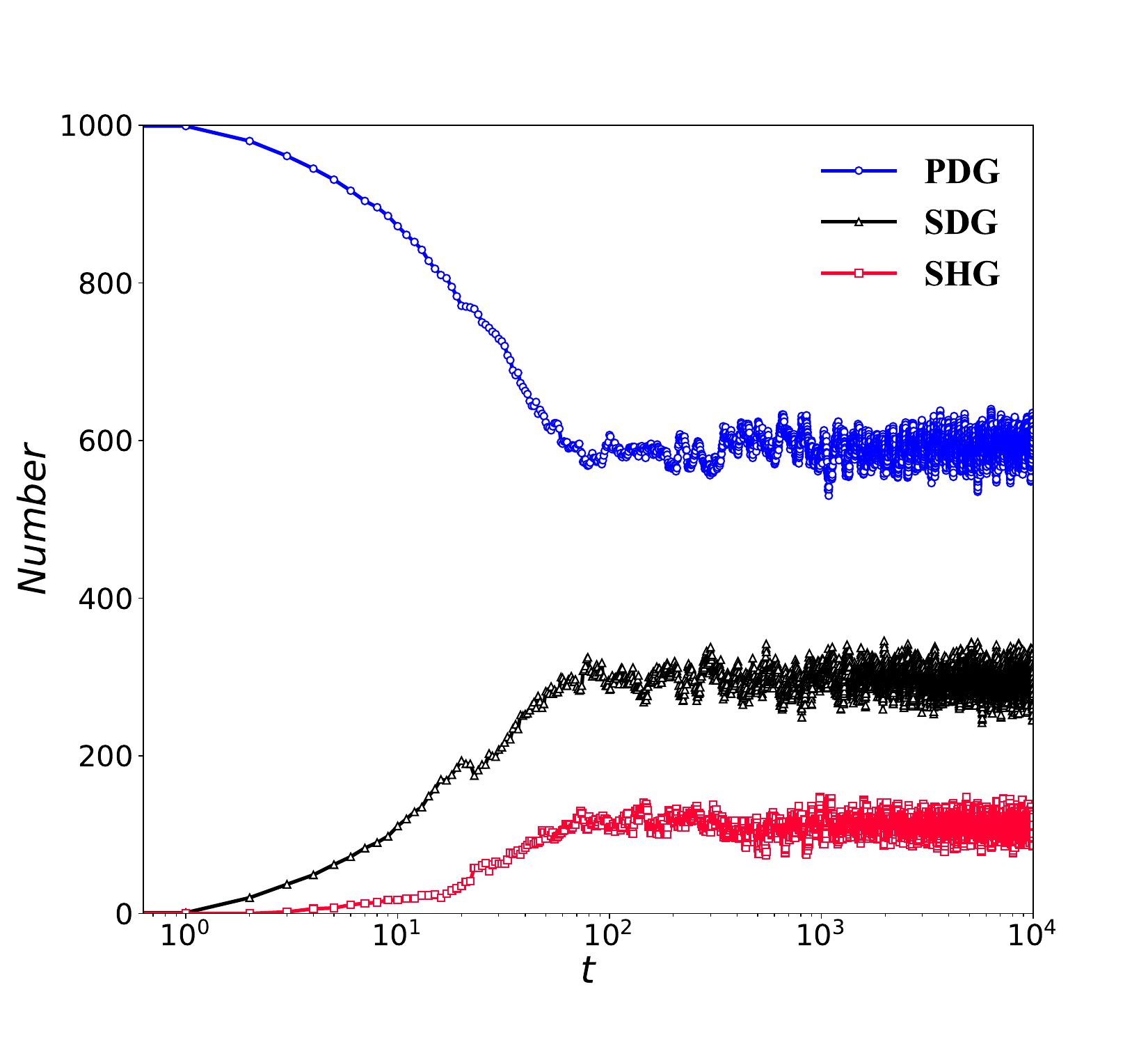}
\label{t_mu2}}
\caption{\textbf{Evolutionary curves of the individual number staying in different game states over time with different parameters.} (a) shows the evolution of the individual number obtained from the initial setting of the parameters with $\lambda_0$=0.02, $\lambda_1$=0.06, $\mu_1$=0.04, $\mu_2$=0.08, and the other subplots are acquired by adjusting a parameter by the control variable method on the initial parameter settings. Specifically, (b) illustrates the individual number with $\lambda_0$ changed to 0.04. (c) shows the individual number with $\lambda_1$ adjusted to 0.12. (d) is the individual number with $\mu_1$ changed to 0.08. And (e) demonstrates the individual number with $\mu_2$ changed to 0.16. As time progresses, the number of individuals in each of the three game states gradually becomes stable.}
\label{t_num}
%\vspace{-1.5\baselineskip}
\end{figure*}
\end{center}

%\vspace{-1.5\baselineskip}
\section{Simulation results and discussions}    \label{part III}

In this section, we will conduct some simulations to verify the model we proposed. In detail, we first illustrate the methods, which will be utilized in the following simulations. Then, we investigate the number of individuals staying in different game states in the second subsection. Subsequently, the evolution of cooperative behavior in the networks is investigated, including the influence of the transition rates between different game states, payoff parameters with and without reputation mechanism, different time scales of strategy updates, and network scale on the cooperation frequency.

%\vspace{-0.5\baselineskip}
\subsection{Methods}

Herein, we explain some methods for our subsequent simulations. For the setting of the payoff matrix in Eq. \ref{Payoff matrix}, as usual, we normalize the advantage of the total payoff of mutual cooperation over that of mutual defection to 2 in all types of social dilemma games by making the reward $R_i = 1$ and the punishment $P_i = 0$. Additionally, if an individual adopts the cooperative strategy, while the other one adopts the defective strategy, the cooperator and the defector receive the $S_i\in[-1, 1]$ and the $T_i\in[0,2]$, respectively. Therefore, based on the relative ordering of $R_i = 1$, $P_i = 0$, $S_i$, and $T_i$, four different types of games can be obtained by dividing the two-dimensional T-S parameter region: the harmony game (HG), the prisoner's dilemma game (PDG), the snowdrift game (SDG), and the stag-hunt game (SHG), of which only the latter three games are social dilemma games, while the HG is not since the dominant strategy is cooperation in this situation. To verify the developed theory, we focus on the game between individuals with three game states, including the PDG, SDG, and SHG, which are three different social dilemma games, and the payoff parameters of the three games are shown in Tab. \ref{payoff parameter}. The transition relationships between the three game types are: the game state of each individual will change from PDG to SDG at a specific rate $\lambda_0$, from SDG to SHG at a rate $\lambda_1$ or to PDG at a rate $\mu_1$ and from SHG to SDG at another specific rate $\mu_2$. The networks in which the individuals are located are a WS small-world network (WS) with 10000 nodes and a square lattice network with periodic boundary (SL) with $200\times200$ nodes, which will be generated by the function $watts\_strogatz\_graph()$ of package $networkx$ in Python and our custom function, respectively. At the initial moment, the strategies of individuals are chosen from cooperate and defect with equal probability and all individuals are in the game state PDG, which means that the number of individuals performing SDG and SHG is 0. Moreover, each individual will be given a random number located in the interval (0, 4) to represent his/her reputation value $Re$, which basically follows the Gaussian distribution $Re\sim G(\mu, \sigma^2)$, where $\mu$ is the mean value and is set to 2, whereas $\sigma$ is the standard deviation and is set to 0.6. Furthermore, the unit $\delta$ of reputation change in Eq. \ref{repuation update} is fixed to 0.04 and the noise factor $\kappa$ in Eq. \ref{probability of updating strategy} is fixed to 0.1 in all simulations. Additionally, we let $Re_j^{t+1}=\frac{\delta}{2}$ if $Re_{j}^{t}-\delta \leq 0$ and let $Re_j^{t+1}=4$ if $Re_{j}^{t}+\delta > 4$, thus trying to ensure that the reputation of the individual is in a reasonable scope. In particular, we emphasize that the time scales of game transitions and strategy updates are different. The evolutionary time of an individual's game state is related to its specific game transition rate, while the time of an individual's strategy update is carried out by following other rules. All simulation evolutionary steps are set to $T = 10^4$, and the final results of each set of parameters are averaged over 5 independent simulations to maintain a good accuracy of the simulation results.

%\vspace{-1.3\baselineskip}
\subsection{The Number of Individuals in the Different Game States}

As mentioned in our model, the game state of each individual changes dynamically during the evolutionary process. We first study the evolution of the number of individuals conducting three different game states in the network under different parameters. As is proved in Thm. \ref{expected thm}, the number of individuals in the three game states is only related to the transition rates ($\lambda_0$, $\lambda_1$, $\mu_1$, and $\mu_2$) but not to the network type. Thus, we do not consider the network type as a variable. Besides, we set the evolution time to be large enough ($T=10^4$) to ensure that the number of individuals in the three game states reaches a stationary level and set the $x$-axis in a logarithmic coordinate to better observe the ascent and descent stage in the evolutionary process.

The evolutionary curves of the individual number staying in different game states over time with different parameters are demonstrated in Fig. \ref{t_num}, where the blue circles, black triangles, and red squares denote the individual number of PDG, SDG, and SHG varying with time, respectively. The transition rates of the game states in Fig. \ref{t_origin} are set to $\lambda_0=0.02$, $\lambda_1=0.06$, $\mu_1=0.04$, and $\mu_2=0.08$, and the rest of the plots are obtained by using the control variable method to change one parameter while keeping the other parameters constant. It can be clearly seen that the individual number in each game state becomes stationary around $t=100$ regardless of the value of transition rates and then fluctuates around a certain value. The subplots in Fig. \ref{t_num} also show that the number of individuals in the PDG tends to gradually decrease and then stabilize, while the number of individuals in the SDG and SHG tends to gradually increase and then stabilize. This is due to the fact that we set all the individuals in the network to be located in PDG at the initial moment, while the number of individuals located in SDG and SHG is 0. The difference in the subplots is the number of individuals located in the three game states at the steady state. Concretely, the individual number in the PDG, SDG, and SHG finally reaches about 530, 268, and 202, respectively in Fig. \ref{t_origin}. Fig. \ref{t_lam0} shows the evolutionary curves with $\lambda_0=0.04$, which is twice larger than that in Fig. \ref{t_origin} and the other parameters ($\lambda_1$, $\mu_1$, and $\mu_2$) remain the same. The individual number of PDG reduces to about 365, while the individual number of SDG and SHG grows to about 364 and 271, which indicates that a larger $\lambda_0$ causes the transition of some individuals located in the PDG to SDG and SHG. By setting the $\lambda_1=0.12$, Fig. \ref{t_lam1} depicts the individual number of PDG, SDG, and SHG eventually fluctuating around 449, 222, and 329, respectively. In Fig. \ref{t_mu1}, where the $\mu_1$ is set to be 0.08, the stationary number of PDG, SDG and SHG becomes 701, 171, and 128, meaning that some individuals will change from SDG and SHG to PDG by increasing $\mu_1$. Fig. \ref{t_mu2} with $\mu_2=0.16$ demonstrates that the number of individuals conducting SHG reduces to approximately 110, while the number of individuals conducting PDG and SDG increases to approximately 593 and 297, which suggests that a larger value of $\mu_2$ results in some individuals performing SHG change to perform PDG and SDG.

\begin{center}
\begin{figure}[htbp]
\centering
\subfigure[PDG]{
\includegraphics[scale=0.13]{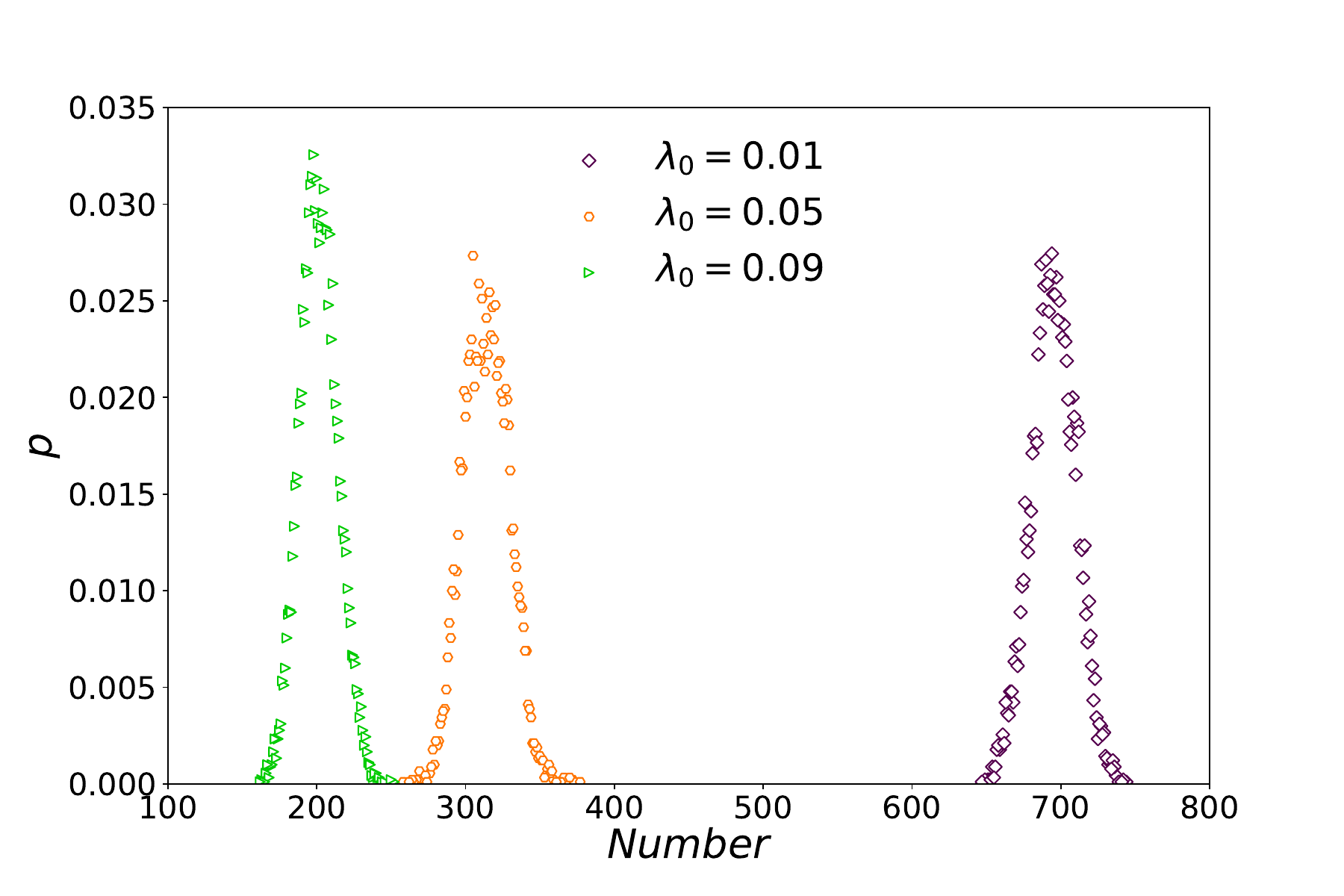}
\label{lam0_PDG}}
\subfigure[SDG]{
\includegraphics[scale=0.13]{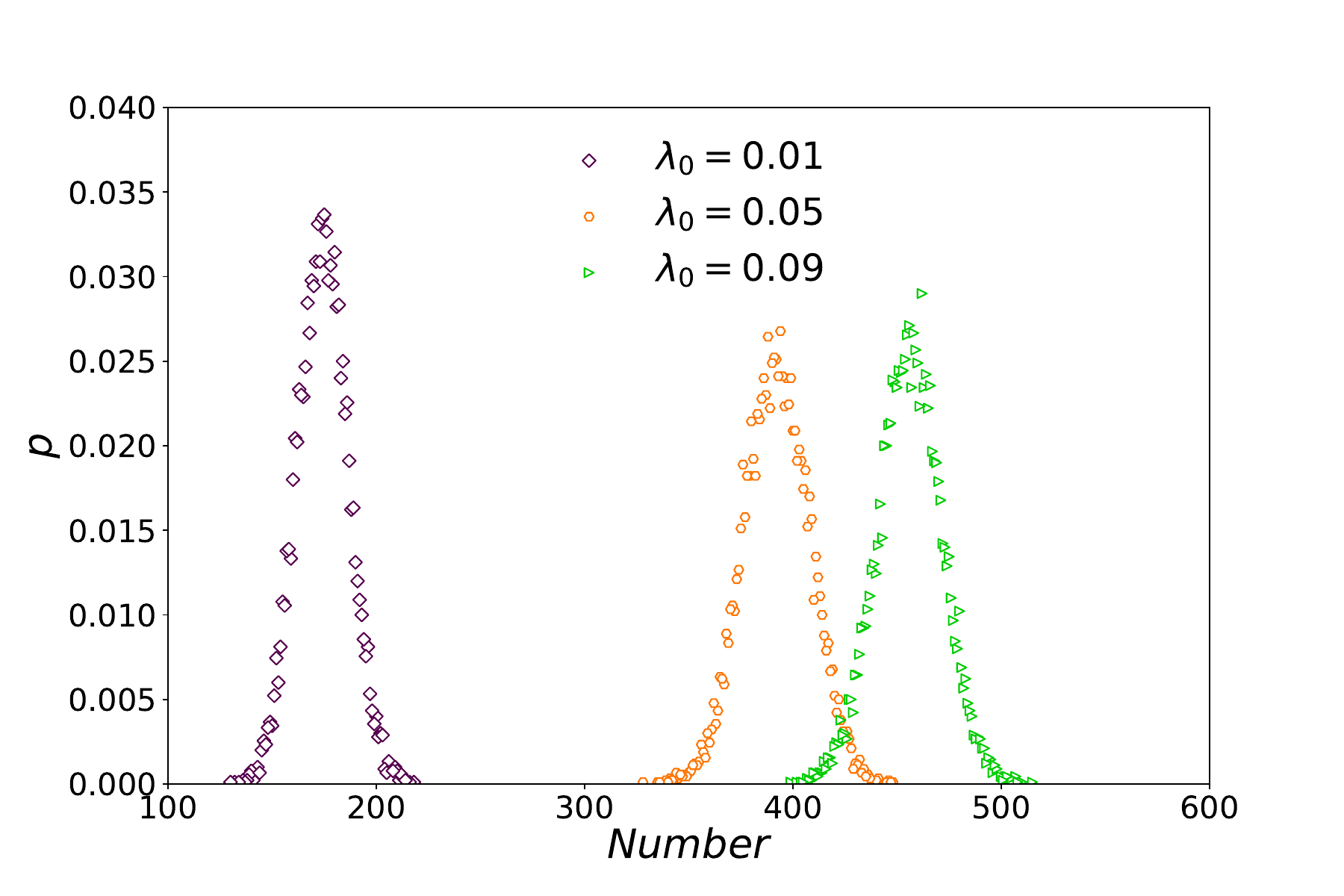}
\label{lam0_SDG}}
\subfigure[SHG]{
\includegraphics[scale=0.13]{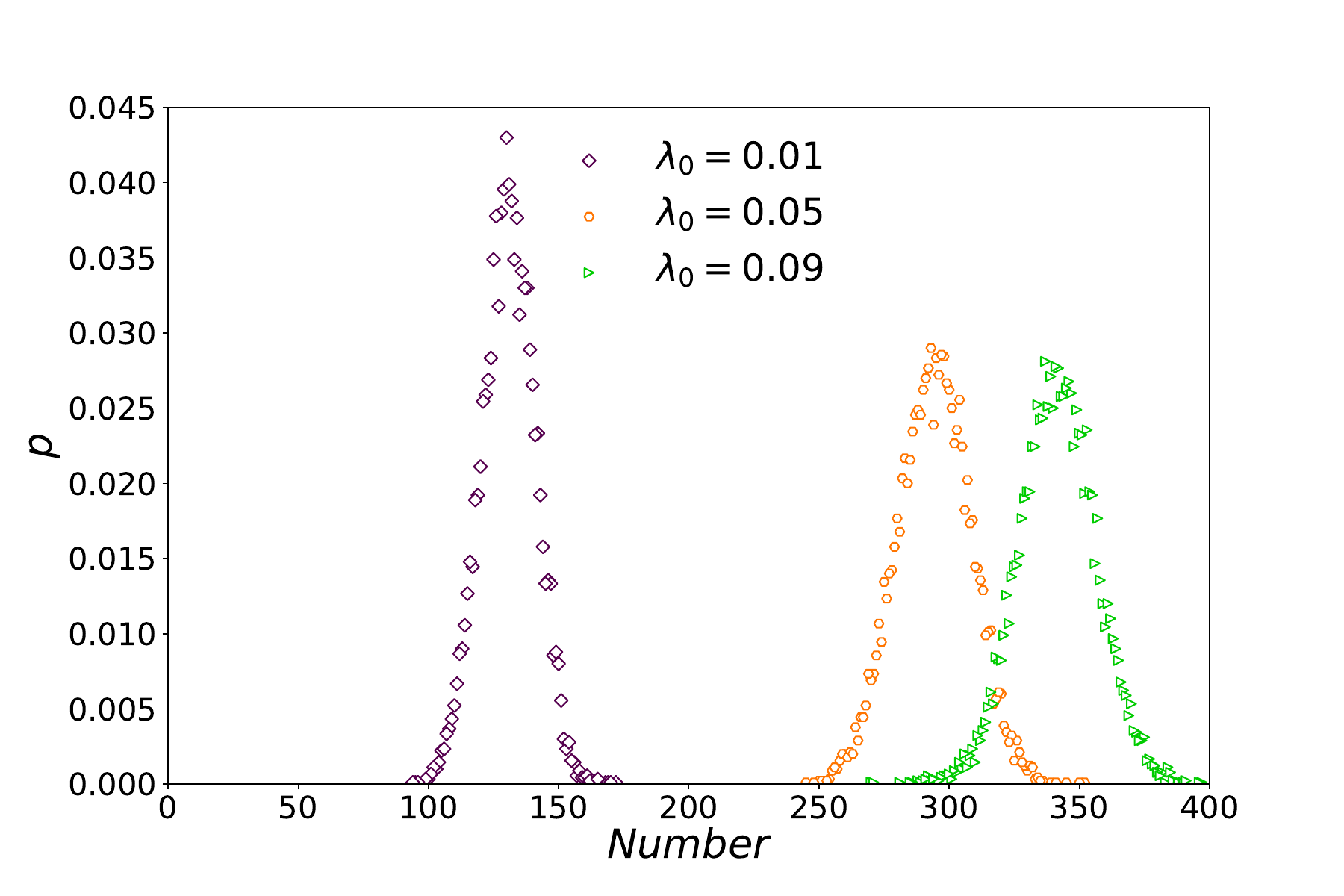}
\label{lam0_SHG}}
\caption{\textbf{Statistical results of the individual number staying in different game states with different $\lambda_0$s.} In this figure, we present the statistical distribution of the individual number in the network in the three game states with different parameters $\lambda_0$s. The $x$-axis and $y$-axis are set as the number of individuals and probability, respectively. (a) is the statistical distribution of the individual number staying in the PDG. (b) displays the statistical distribution of the individual number located in the SDG. And (c) demonstrates the statistical distribution of the individual number located in the SHG.}
\label{lam0_num}
\end{figure}
\end{center}

\begin{center}
\begin{figure}[htbp]
\centering
\subfigure[PDG]{
\includegraphics[scale=0.13]{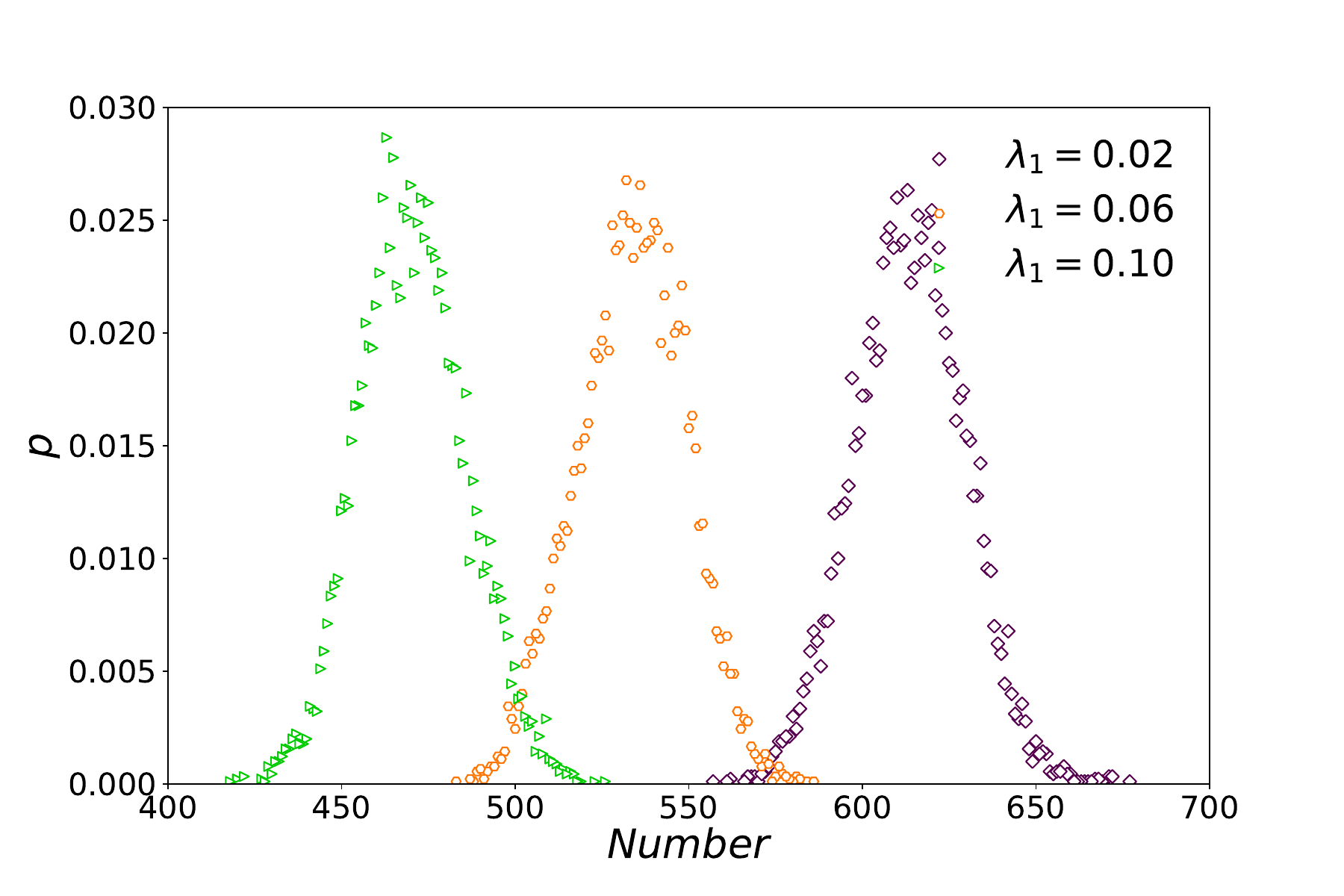}
\label{lam1_PDG}}
\subfigure[SDG]{
\includegraphics[scale=0.13]{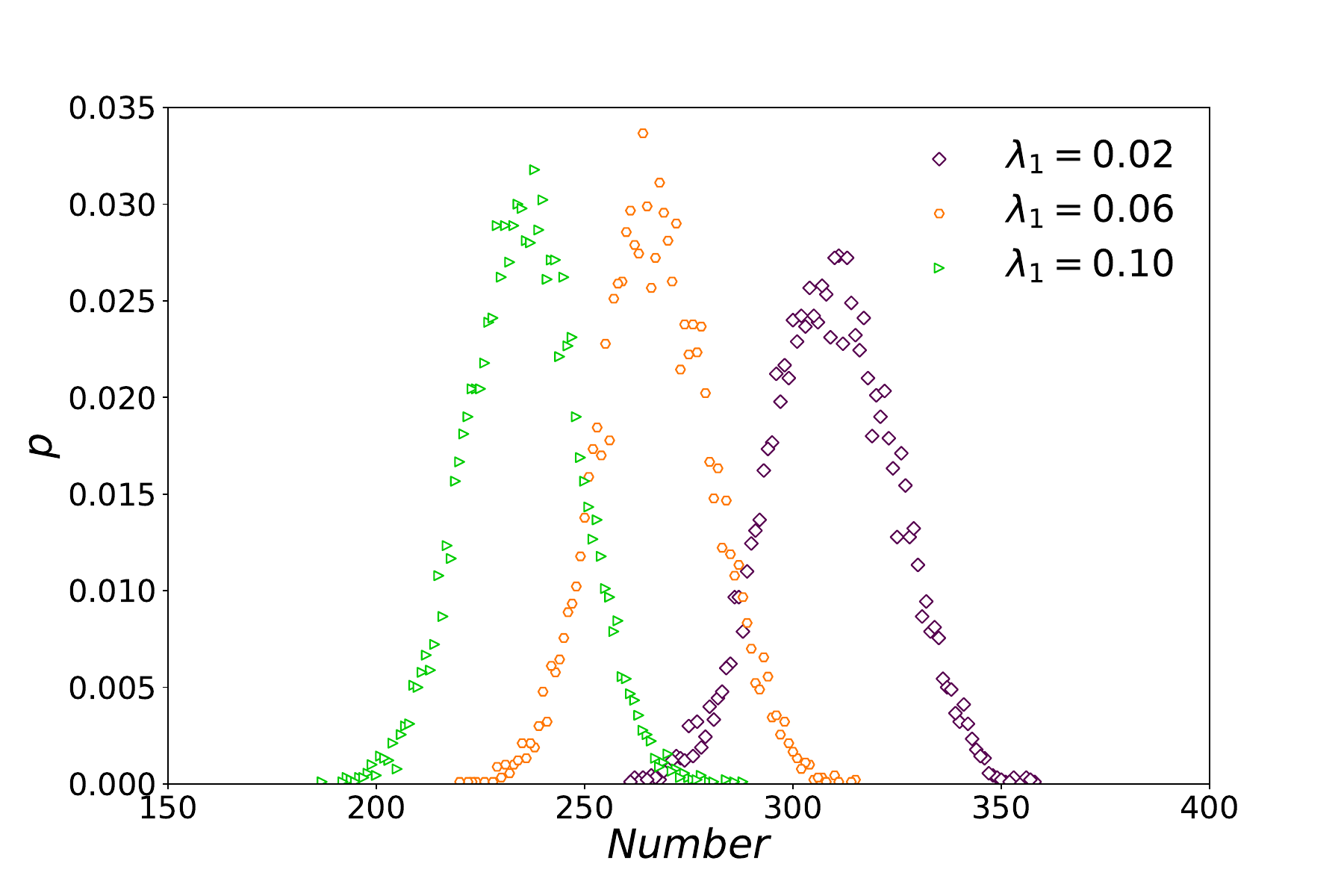}
\label{lam1_SDG}}
\subfigure[SHG]{
\includegraphics[scale=0.13]{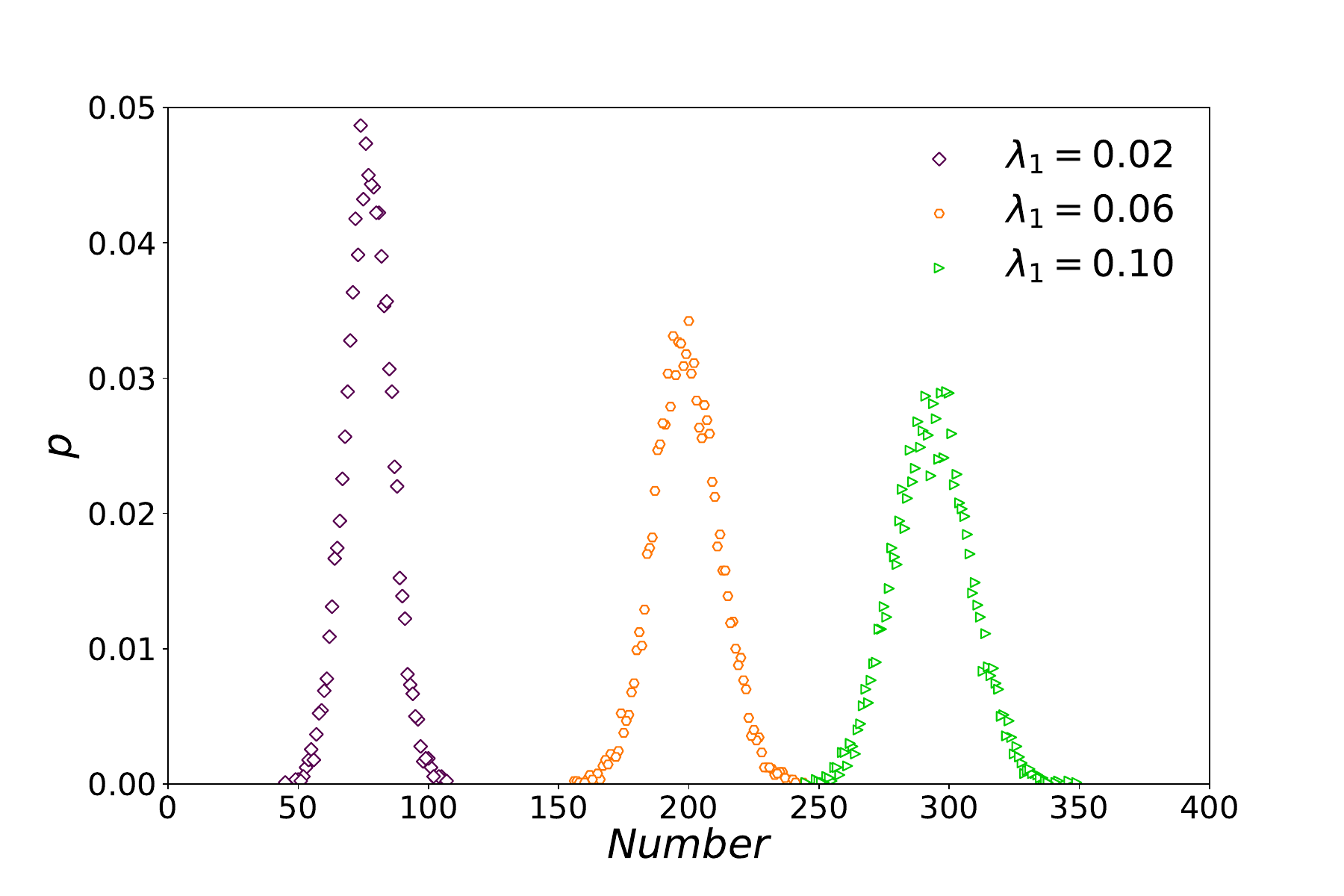}
\label{lam1_SHG}}
\caption{\textbf{Statistical results of the individual number staying in different game states with different $\lambda_1$s.} In this figure, we show the statistical distribution of the individual number in the network in the three game states with different parameters $\lambda_1$s. The $x$-axis and $y$-axis are set as the number of individuals and probability, respectively. (a) is the statistical distribution of the individual number staying in the PDG. (b) displays the statistical distribution of the individual number located in the SDG. And (c) demonstrates the statistical distribution of the individual number located in the SHG.}
\label{lam1_num}
\end{figure}
\end{center}

\begin{center}
\begin{table*}[htbp]
\renewcommand{\arraystretch}{1.5}
\begin{center}
\caption{The results of the stationary number of individuals staying in different game states with $\lambda_1=0.06,\mu_1=0.04,\mu_2=0.08$}
\begin{tabular}{c|ccc|ccc|ccc}
\hline
\multirow{2}{*}{Results} & \multicolumn{3}{c|}{$\lambda_0=0.01$} & \multicolumn{3}{c|}{$\lambda_0=0.05$} & \multicolumn{3}{c}{$\lambda_0=0.09$} \\ \cline{2-10}
                         & PDG    & SDG    & SHG   & PDG    & SDG    & SHG   & PDG    & SDG    & SHG   \\ \hline
Theoretical results      & 695.652    & 173.913    & 130.435   & 313.725      & 392.156      & 294.118     & 202.532      & 455.696     & 341.772     \\
Simulation results       & 695.081    & 174.303    & 130.611   & 313.900      & 392.101     & 293.999     & 201.954      & 456.265     & 341.781     \\
Relative error           & 0.082\%      & 0.224\%      & 0.135\%     & 0.056\%      &  0.014\%     & 0.041\%    & 0.285\%      & 0.125\%     & 0.003\% \\
\hline
\end{tabular}
\label{statistics with lam0}
\end{center}
\end{table*}
\end{center}

\begin{center}
\begin{table*}[htbp]
\renewcommand{\arraystretch}{1.5}
\begin{center}
\caption{The results of the stationary number of individuals staying in different game states with $\lambda_0=0.02,\mu_1=0.04,\mu_2=0.08$}
\begin{tabular}{c|ccc|ccc|ccc}
\hline
\multirow{2}{*}{Results} & \multicolumn{3}{c|}{$\lambda_1=0.02$} & \multicolumn{3}{c|}{$\lambda_1=0.06$} & \multicolumn{3}{c}{$\lambda_1=0.10$} \\ \cline{2-10}
                         & PDG    & SDG    & SHG   & PDG    & SDG    & SHG   & PDG    & SDG    & SHG   \\ \hline
Theoretical results      & 615.385    & 307.692    & 76.923   & 533.333      & 266.667      & 200.000     & 470.588      & 235.294     & 294.118     \\
Simulation results       & 613.697    & 309.382    & 76.921   & 533.699      & 267.122     & 199.179     & 470.813      & 235.520    & 293.668     \\
Relative error           & 0.274\%      & 0.549\%      & 0.003\%     & 0.069\%      &  0.171\%     & 0.411\%    & 0.048\%      & 0.096\%     & 0.153\% \\
\hline
\end{tabular}
\label{statistics with lam1}
\end{center}
\end{table*}
\end{center}
%\vspace{-0.8\baselineskip}

%\vspace{-2.5\baselineskip}
According to our theory, the number of PDG, SDG, and SHG performed in the network should be $n\mu_1\mu_2/(\mu_1\mu_2+\lambda_0\mu_2+\lambda_0\lambda_1)$, $n\lambda_0\mu_2/(\mu_1\mu_2+\lambda_0\mu_2+\lambda_0\lambda_1)$, and $n\lambda_0\lambda_1/(\mu_1\mu_2+\lambda_0\mu_2+\lambda_0\lambda_1)$ when the stationary state is reached, which indicates that the number of individuals in a certain game state is only related to the size of the network and the transition rates of the game state. Subsequently, we show the simulated scale distributions of the number of individuals staying in three game states. The results are shown in Figs. \ref{lam0_num} and \ref{lam1_num}, where $\lambda_0$ and $\lambda_1$ are employed as independent variables, respectively. Concretely, we record the number of individuals in three different game states at the last 9000 steps of a total of $10^4$ steps, and the number of individuals has evolved stably (as can be seen in Fig. \ref{t_num}). Next, we utilize the function $Counter$ of package $collections$ in Python to count the frequency of each number of individuals and treat the frequency as probability according to the law of large numbers. By setting the parameters $\lambda_1=0.06,\mu_1=0.04,\mu_2=0.08$, we demonstrate the statistical distributions of the individual number staying in PDG, SDG, and SHG with different $\lambda_0$s ($\lambda_0=0.01,0.05$, and $0.09$) in Figs. \ref{lam0_PDG}, \ref{lam0_SDG}, and \ref{lam0_SHG}, from which we get that each probability distribution approximately follows a normal distribution. Besides, in Figs. \ref{lam0_PDG}, \ref{lam0_SDG}, and \ref{lam0_SHG}, for $\lambda_0=(0.01, 0.05, 0.09)$, the most probable number of individuals staying in PDG, SDG, and SHG are approximately (695, 314, 202), (174, 392, 456), and (131, 294, 342), respectively. It can be clearly seen that the distributions in Fig. \ref{lam0_num} are narrow bands, i.e., the deviations are small.

We also compare our theoretical results with the simulation results in Tab. \ref{statistics with lam0} to further verify our theory. The relative error is calculated by $e=\left| x-x^* \right|/x$, where $x$ is the theoretical result and $x^*$ is the simulation one. The maximum value of the relative error in Tab. \ref{statistics with lam0} equals 0.285\%, which means that the results obtained from the simulation are very close to the theoretical results and demonstrate the validity of our theoretical analysis in Thm. \ref{expected thm}. Furthermore, we can see that the number of individuals staying in PDG will reduce as the transition rate $\lambda_0$ grows, while the number of individuals located in SDG and SHG will increase.

Next, we show the statistical distributions of the individual number staying in PDG, SDG, and SHG with different values $\lambda_1$ ($\lambda_1=0.02, 0.06$, and $0.10$) in Figs. \ref{lam1_PDG}, \ref{lam1_SDG}, and \ref{lam1_SHG}, where the other parameters are set as $\lambda_0=0.02,\mu_1=0.04$, and $\mu_2=0.08$. We see that the horizontal coordinates corresponding to the peaks of individuals located in the same game state are different for different $\lambda_1$. In detail, in Fig. \ref{lam1_PDG}, the most probable number of individuals staying in PDG with $\lambda_1=0.02$ marked by the purple diamond, $\lambda_1=0.06$ marked by the orange circle, and $\lambda_1=0.10$ marked by the green triangle are about 615, 533, and 471, respectively. For Figs. \ref{lam1_SDG} and \ref{lam1_SHG}, the horizontal coordinates corresponding to the peaks of individuals located in SDG and SHG are around (308, 267, 235) and (77, 200, 294), respectively. Analogously, we calculate the relative error to further verify our theory, and the results are shown in Tab. \ref{statistics with lam1}, from which we infer that all the relative errors are less than 0.549\%, which also illustrates the accuracy of our theoretical analysis in Thm. \ref{expected thm}. Furthermore, we obtain that the distributions of Fig. \ref{lam1_num} are wider than that of Fig. \ref{lam0_num}, which suggests that the deviations of Fig. \ref{lam1_num} are larger. Fig. \ref{lam1_num} also exhibits that the number of individuals staying in PDG and SDG will decrease as the transition rate $\lambda_1$ increases, while the number of individuals located in SHG will grow.

%\vspace{-0.5\baselineskip}
\subsection{The Effect of $\mu_1$ and $\mu_2$ on the Cooperation Density}
\label{effect of mu1 and mu2}

In this simulation, we investigate the influence of the transition rates $\mu_1$ from SDG to PDG and $\mu_2$ from SHG to SDG on the cooperation of WS and SL networks. Except $\mu_1$ and $\mu_2$, by fixing the parameters $b=1.5$, $r=0.5$, and $\lambda_0 = \lambda_1 = 0.03$, we show the function of the cooperation frequency ($f_c$) on $\mu_1$ under different $\mu_2$ in Fig. \ref{mu1_f(mu2)}.

%\vspace{-1.9\baselineskip}
\begin{center}
\begin{figure}[htbp]
\centering
\subfigure[WS]{
\includegraphics[scale=0.14]{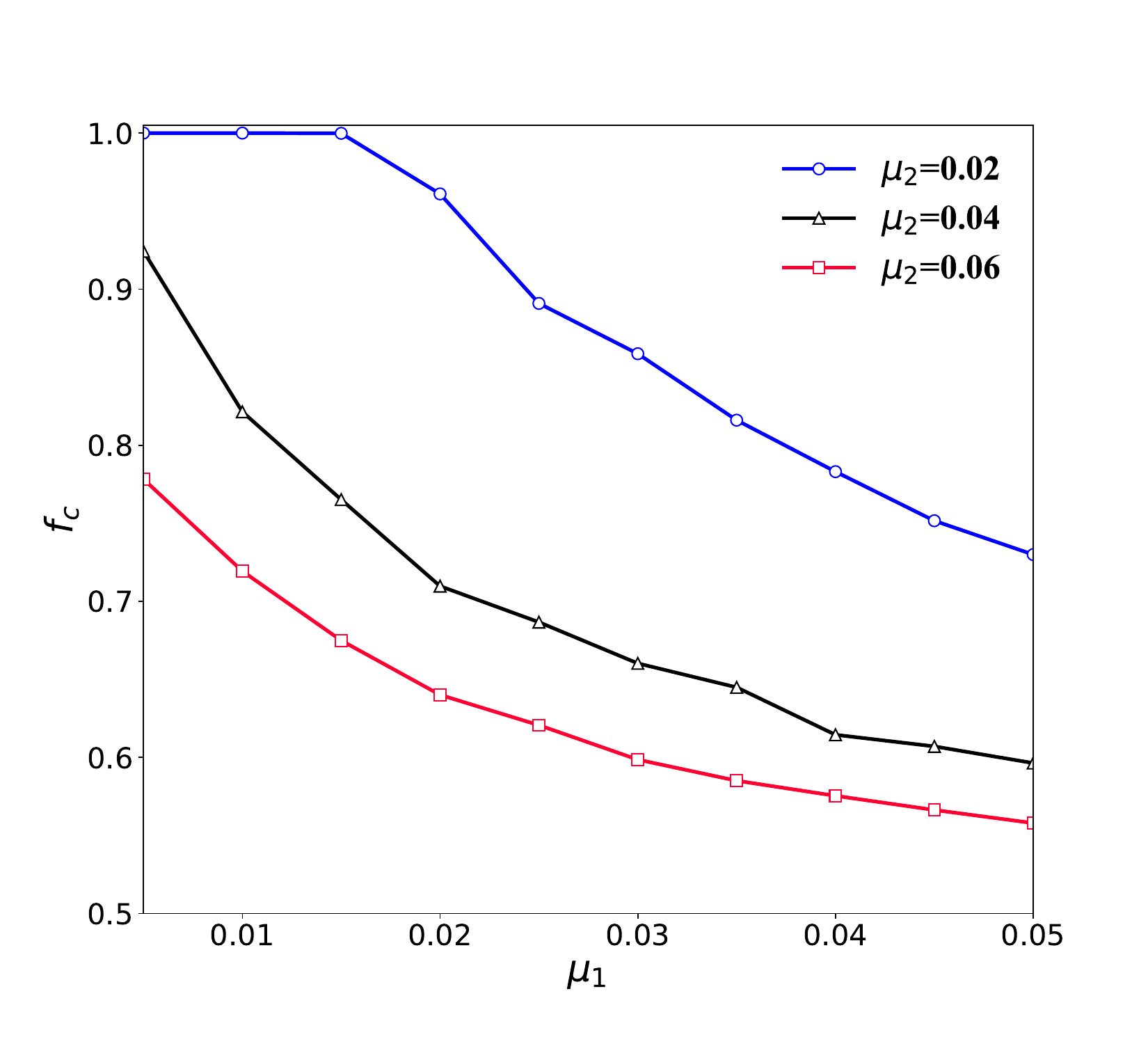}
\label{mu1_f(a)}}
%\hspace{0.5cm}
\subfigure[SL]{
\includegraphics[scale=0.14]{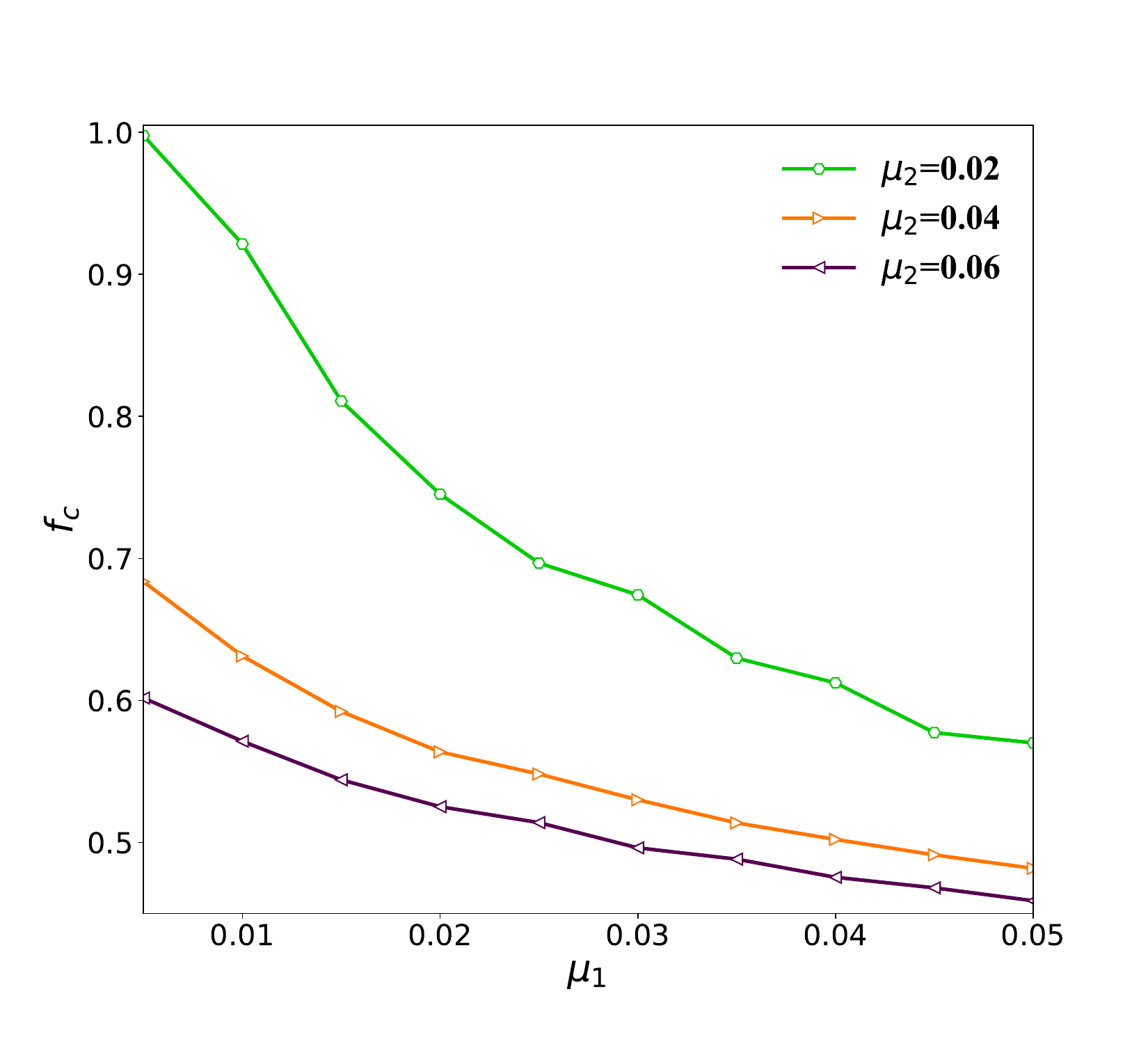}
\label{mu1_f(b)}}
\caption{\textbf{Plots of cooperation density against $\mu_1$ under different $\mu_2$.} By setting the payoff parameters $b=1.5$ of PDG and $r=0.5$ of SDG and SHG, we present the cooperation density against $\mu_1$ under different $\mu_2$s on the WS (in subplot (a)) and SL (in subplot(b)) networks, respectively, where the $x$-axis is set as $\mu_1$, which denotes the transition rate from SDG to PDG, while the $y$-axis is set as the cooperation density. The range of each $x$-axis is set as [0.005, 0.05], whereas the ranges of $y$-axis of subplot(a) and subplot(b) are set as [0.5, 1] and [0.45, 1], respectively. In general, it can be seen that the increases of both $\mu_1$ and $\mu_2$ inhibit the cooperation density on the WS and SL networks.}
\label{mu1_f(mu2)}
\end{figure}
\end{center}

%%\vspace{-1.9\baselineskip}
As shown both in Figs. \ref{mu1_f(a)} and \ref{mu1_f(b)}, we yield that the frequency of cooperation decreases as $\mu_1$ and $\mu_2$ increase on both the WS and SL networks, i.e., both $\mu_1$ and $\mu_2$ act as a disincentive to the cooperative behavior of the networks. Moreover, under the same conditions (the same $\mu_1$ and $\mu_2$), the proportion of cooperation on the WS network in Fig. \ref{mu1_f(a)} is higher than that on the SL network in Fig. \ref{mu1_f(b)}, which implies that the WS network facilitates the evolution of cooperation more than that of the SL network. Furthermore, within the same range of parameter $\mu_1$, different $\mu_2$ have different degrees of decline. For example, on the SL network in Fig. \ref{mu1_f(b)}, the cooperation frequency corresponding to $\mu_2=0.02$ decreases from 1 at the beginning to 0.57 at the end, the cooperation frequency corresponding to $\mu_2=0.04$ decreases from 0.68 to 0.48, while the cooperation frequency corresponding to $\mu_2=0.06$ decreases from 0.60 to 0.46, and the degree of decrease in these three cases are 0.43, 0.20, and 0.14, respectively. Therefore, we can derive that the degree of decline of $f_c$ decreases with increasing $\mu_2$.

Subsequently, we make an explanation of the above phenomenon in Fig. \ref{mu1_f(mu2)}. Primarily, the three game models have different Nash equilibria, with (D, D) for PDG, (D, C) for SDG, and (D, D) or (C, C) for SHG. When the rate of transition from SDG to PDG is greater than the rate of transition from PDG to SDG, i.e., $\mu_1 > \lambda_0$, there will be more individuals in the network who tend to choose defect. In a similar way, some individuals in the network will change from cooperators to defectors when the rate of transformation of SHG to SDG is greater than the rate of transformation of SDG to SHG, i.e., $\mu_2 > \lambda_1$. Besides, a larger difference in the transition rate will lead to the above phenomenon being exacerbated, namely, the percentage of cooperation in the network becomes lower. Both $\lambda_0$ and $\lambda_1$ are fixed in the simulations of Fig. \ref{mu1_f(mu2)}, so the cooperative behavior will be suppressed with the increase of $\mu_1$ and $\mu_2$.

%\vspace{-0.5\baselineskip}
\subsection{The Influence of $\lambda_0$ and $\lambda_1$ on the Cooperation Frequency}

\begin{center}
\begin{figure}[htbp]
\centering
\subfigure[WS]{
\includegraphics[scale=0.27]{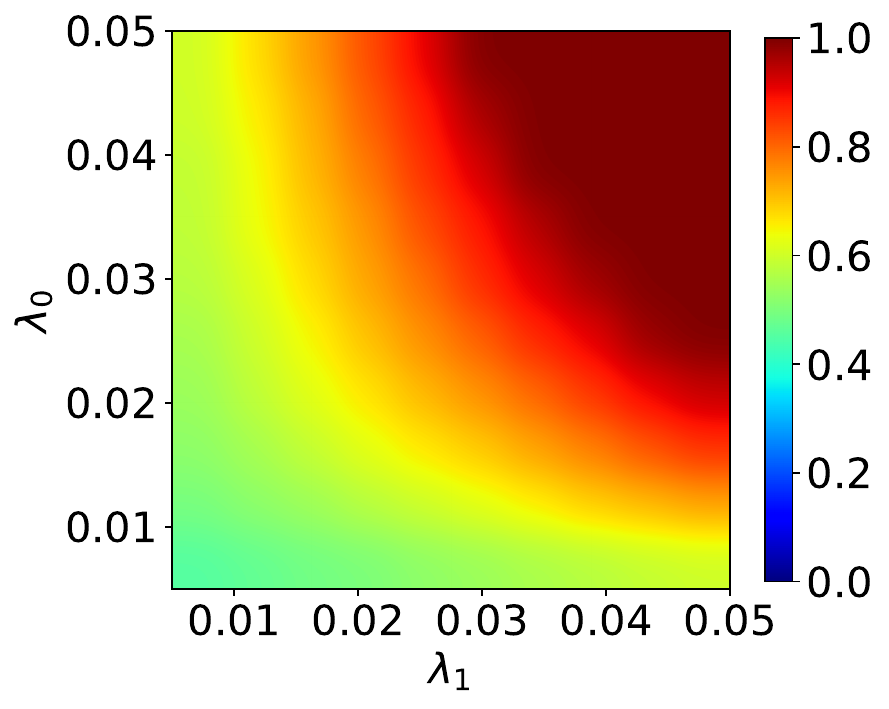}
\label{lam0_lam1_f(a)}}
%\hspace{1cm}
\subfigure[SL]{
\includegraphics[scale=0.27]{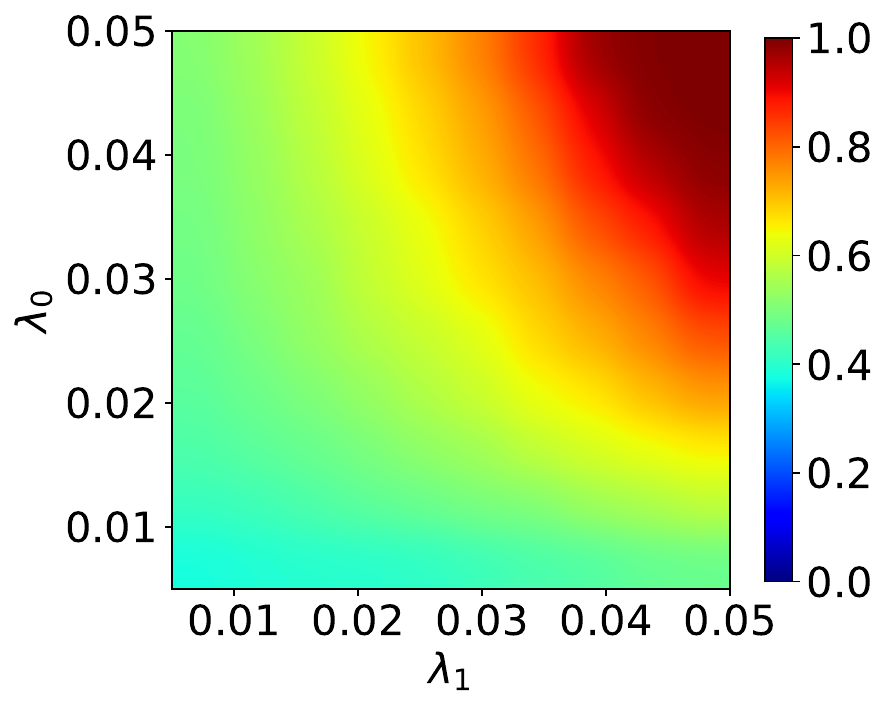}
\label{lam0_lam1_f(b)}}
\caption{\textbf{Heat maps of cooperation frequency with respect to parameters $\lambda_0$ and $\lambda_1$.} This figure demonstrates the influence of $\lambda_0$ and $\lambda_1$ on the cooperation frequency on the WS (in panel (a)) and SL (in panel (b)) networks. We set the payoff parameters $b=1.5$ of PDG and $r=0.5$ of SDG and SHG. The $x$-axis is set as the $\lambda_1$ with the range [0.005, 0.05], which represents the transition rate from SDG to SHG, and the $y$-axis is set as the $\lambda_0$ with the same range, which indicates the transition rate from PDG to SDG. Both WS and SL networks demonstrate that not only parameter $\lambda_0$ but also parameter $\lambda_1$ can facilitate the emergence of cooperation.}
\label{lam0_lam1_f}
\end{figure}
\end{center}

%\vspace{-1.5\baselineskip}
As shown in subsection \ref{effect of mu1 and mu2}, the transition rates $\mu_1$ and $\mu_2$ have a large effect on the cooperative behavior, and in this subsection, we will explore the influence of $\lambda_0$ and $\lambda_1$ on the evolution of cooperation on WS and SL networks. By fixing the payoff parameters $b=1.5$ and $r=0.5$, and the other two transition rates $\mu_1=0.03$ and $\mu_2=0.02$, we present the heat maps of the cooperation frequency with respect to the parameters $\lambda_0$ and $\lambda_1$ in Fig. \ref{lam0_lam1_f}, where a warmer color (the color is closer to red) means a higher percentage of cooperation. The cooperation frequency of each parameter pair ($\lambda_1$, $\lambda_0$) is averaged by 5 independent simulations, and the cooperation frequency in each simulation is gained by averaging the last 500 steps of cooperation density in the total of $10^4$ evolution steps. In Fig. \ref{lam0_lam1_f(a)}, we see that the color changes from cool (blue) to warm (red) as $\lambda_0$ and $\lambda_1$ grow, which indicates that the cooperation ratio is enhanced. Additionally, pure cooperators will appear on the WS network when both $\lambda_0$ and $\lambda_1$ are large (e.g., $\lambda_0 > 0.035$ and $\lambda_1 > 0.035$), while almost no pure defectors will emerge in the network no matter how small $\lambda_0$ and $\lambda_1$ are. In Fig. \ref{lam0_lam1_f(b)}, the SL network also exhibits an overall increase in the proportion of cooperation with increasing parameters $\lambda_0$ and $\lambda_1$, which is consistent with the phenomenon on the WS network. However, what is different from the WS network is the presence of many defectors on the SL network, which appears when $\lambda_0$ is very small. The region where pure cooperators appear on the SL network is also narrower than that on the WS network and only appears when both $\lambda_0$ and $\lambda_1$ are very large (e.g., $\lambda_0 > 0.045$ and $\lambda_1 > 0.045$).

%\vspace{-0.12\baselineskip}
As explained in subsection \ref{effect of mu1 and mu2}, SDG facilitates the emergence of cooperation more than PDG, but less than SHG for the evolution of cooperation. The parameters $\lambda_0$ and $\lambda_1$ mean the transition rate from PDG to SDG and the transition rate from SDG to SHG, respectively. Therefore, increasing both parameters can facilitate the transformation of the game model towards a game model that is more conducive to the emergence of cooperation. In addition, we obtain that the WS network is more favorable to the survival of cooperators than that of the SL network by comparing Fig. \ref{lam0_lam1_f(a)} with Fig. \ref{lam0_lam1_f(b)}, which is in accordance with our previous analysis.

\begin{center}
\begin{figure}[htbp]
\centering
\subfigure[WS]{
\includegraphics[scale=0.27]{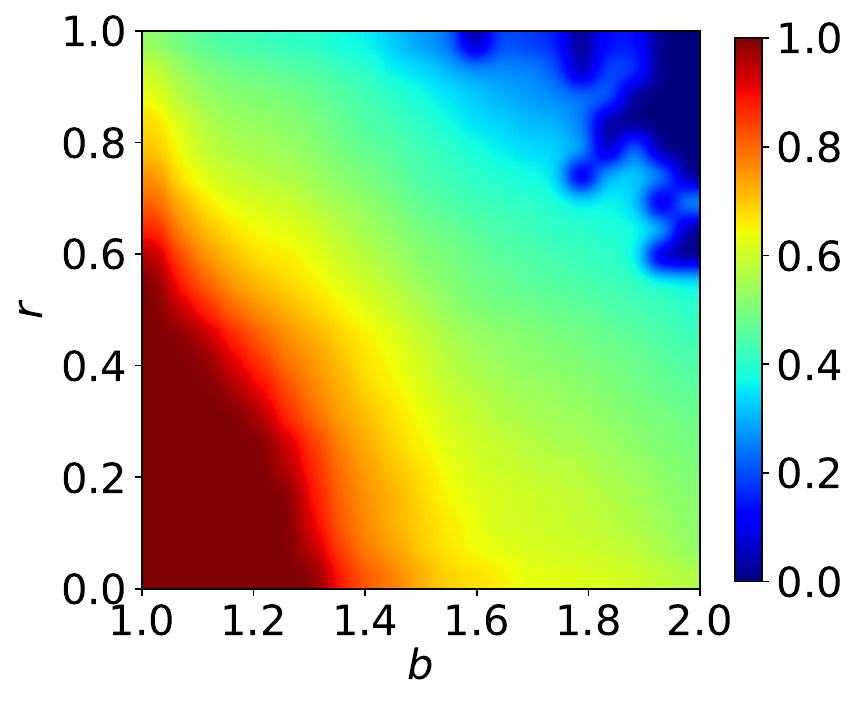}
\label{r_b_f(a)}}
%\hspace{1cm}
\subfigure[SL]{
\includegraphics[scale=0.27]{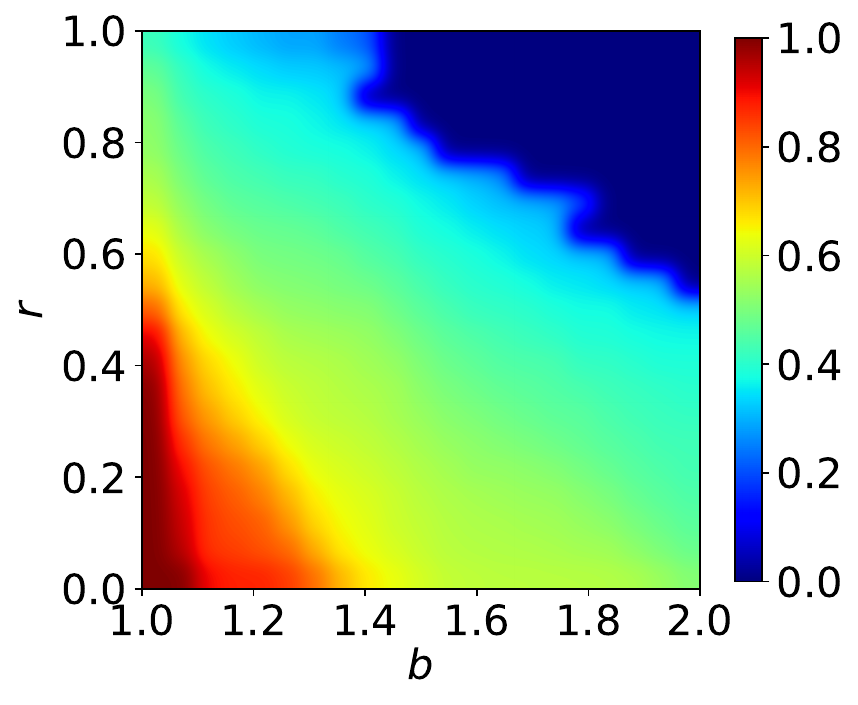}
\label{r_b_f(b)}}
\caption{\textbf{Heat maps of cooperation density about parameters $b$ and $r$.} By setting the $x$-axis as the payoff parameter $b$ with range [1, 2] and $y$-axis as the payoff parameter $r$ with range [0, 1], we show the heat maps of cooperation density on WS (in subplot(a)) and SL (in subplot(b)) networks with respect to payoff $b$ and $r$, from which we can obtain that the cooperation will be promoted by decreasing the payoff parameter $b$ or $r$.}
\label{r_b_f}
\end{figure}
\end{center}

%\vspace{-2.5\baselineskip}
\subsection{The Effect of Payoff Parameters on Network Cooperation Behavior}

In our previous study, we focused on the effect of the game state transition rate on the cooperative behavior of the network, while fixing the payoff parameters $b$ and $r$. Herein, we will investigate the effect of the payoff parameters $b$ and $r$ on the evolution of network cooperation by fixing the other parameters except for the payoff parameters. Specifically, we set the four transition rates between PDG, SDG, and SHG as $\lambda_0=0.015$, $\lambda_1=0.01$, $\mu_1=0.03$, and $\mu_2=0.02$. The results of the evolution of the cooperation ratio about $b$ and $r$ on the WS and SL networks are shown in Figs. \ref{r_b_f(a)} and \ref{r_b_f(b)}, respectively. In the WS network, pure cooperators will emerge with small $b$ and $r$ (e.g., $b < 1.2$ and $r < 0.3$), while pure defectors will appear with very large $b$ and $r$ (e.g., $b > 1.9$ and $r > 0.8$), which can be seen from Fig. \ref{r_b_f(a)}. While in the SL network, unlike the WS network, pure cooperators will emerge with very small $b$ and $r$ (e.g., $b < 1.05$ and $r < 0.05$), while pure defectors will appear with large $b$ and $r$ (e.g., $b > 1.7$ and $r > 0.7$), which can be seen from Fig. \ref{r_b_f(b)}. Besides, we can clearly see that the region of pure cooperators on the WS network is larger than that on the SL network, while the area of pure defectors is smaller than that on the SL network, which indirectly proves that the WS network is more beneficial to the emergence of cooperators than the SL network. We also note that the color distribution in Fig. \ref{r_b_f} is not strictly ordered, which is likely due to the randomness of the initial setup and the players deciding their next strategy based on the previous step.

Next, we give the reasons for the phenomenon arising in Fig. \ref{r_b_f}. Primarily, we infer that the payoff parameter $b$ is only related to the PDG and the payoff of a defector performing PDG in the network will increase as $b$ grows according to the parameters of the weak prisoner's dilemma game in Tab. \ref{payoff parameter}, leading to individuals who are engaged in PDG being more likely to choose to be defector than cooperator. Analogously, for the payoff parameter $r$, which is related to both SDG and SHG, the payoff of a defector conducting SDG and SHG in the network will rise as $r$ increases, while the payoff of a cooperator will reduce, resulting in the defective strategy becoming the preferred strategy for individuals who perform SDG and SHG in the network. Therefore, we can conclude that the growth of both the payoff parameters $b$ and $r$ will have an inhibitory influence on the emergence of cooperative behavior in the network.

%\vspace{-0.8\baselineskip}
\subsection{The Impact of Payoff Parameters on Network Cooperation Behavior without Reputation Mechanism}

It is worth noting that our previous results show the effectiveness of the proposed model on the evolution of cooperation, but these results depend on the game transition as well as the reputation mechanism. In this subsection, we explore the evolutionary behavior of cooperation in such cases by only taking the game transition into account. The results of the cooperative evolution on the WS and SL networks are respectively shown in Figs. \ref{r_b_f(a)_noreputation} and \ref{r_b_f(b)_noreputation}, with exactly the same parameter settings as in Figs. \ref{r_b_f(a)} and \ref{r_b_f(b)}, except that there is no reputation mechanism.

\begin{center}
\begin{figure}[htbp]
\centering
\subfigure[WS]{
\includegraphics[scale=0.27]{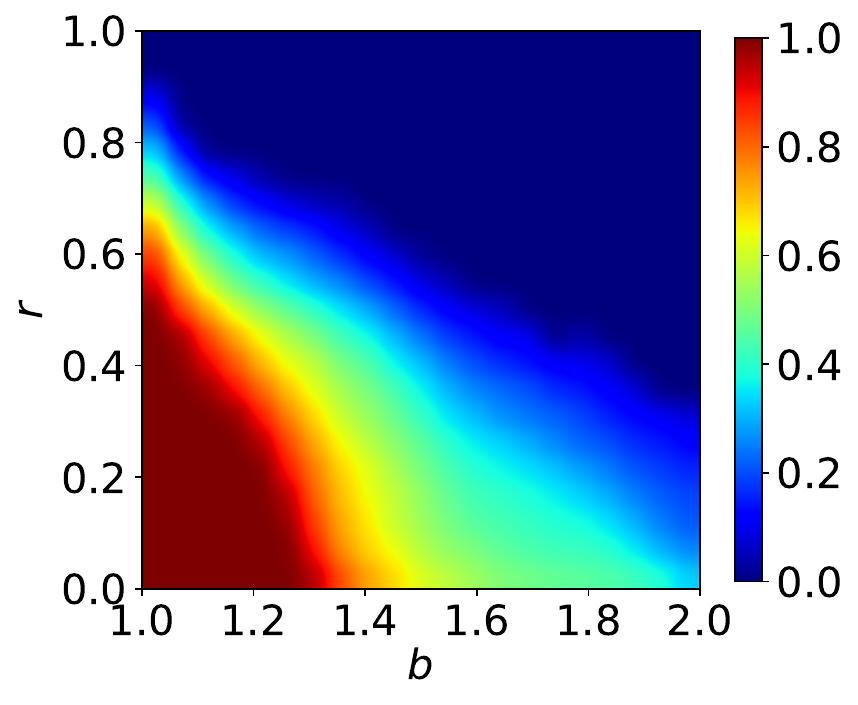}
\label{r_b_f(a)_noreputation}}
%\hspace{1cm}
\subfigure[SL]{
\includegraphics[scale=0.27]{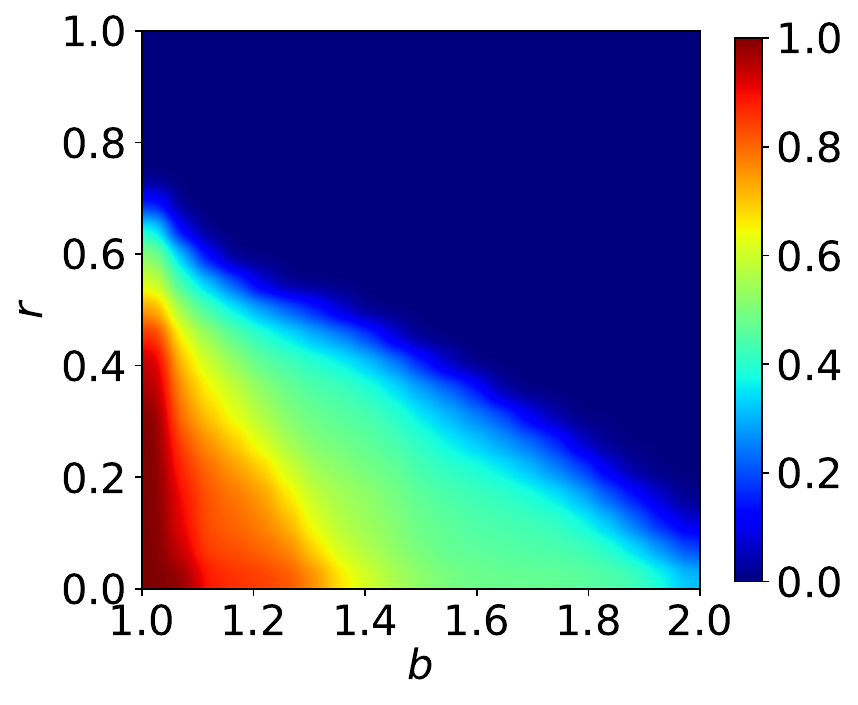}
\label{r_b_f(b)_noreputation}}
\caption{\textbf{Heat maps of cooperators about parameters $b$ and $r$ without reputation mechanism.} This figure displays the influence of $b$ and $r$ on the cooperation density on the WS (in panel (a)) and SL (in panel (b)) networks without reputation mechanism and the parameter settings are exactly the same as in Fig. \ref{r_b_f}. Both networks exhibit a decrease in the number of cooperators who consider only the game transitions compared to the combined consideration of reputation mechanisms.}
\label{r_b_f_noreputation}
\end{figure}
\end{center}

%\vspace{-0.18\baselineskip}
We observe that the proportion of cooperators still decreases as $b$ and $r$ increase, whether on the WS or SL networks. Plus, we can also see that the area of dark red in Fig. \ref{r_b_f(a)_noreputation} is much larger than that in Fig. \ref{r_b_f(b)_noreputation}, while the area of dark blue in Fig. \ref{r_b_f(b)_noreputation} is larger than that in Fig. \ref{r_b_f(a)_noreputation}, which indicates that the WS network is more conducive to the emergence of cooperators than the SL network even without the reputation mechanism. In addition, by comparing Fig. \ref{r_b_f(a)} with Fig. \ref{r_b_f(a)_noreputation} and Fig. \ref{r_b_f(b)} with Fig. \ref{r_b_f(b)_noreputation}, we can find that although their conditions for the emergence of pure cooperators are similar, the circumstances for the emergence of pure defectors on Figs. \ref{r_b_f(a)_noreputation} and \ref{r_b_f(b)_noreputation} are much greater than those on Figs. \ref{r_b_f(a)} and \ref{r_b_f(b)}, which can be gained from the areas of dark red and dark blue in the four subfigures. Therefore, through this comparative experiment, we get that although only taking the game transition into account can promote the evolution of cooperation, the promotive effect is weaker than that of both the reputation mechanism.

%\vspace{-1.5\baselineskip}
\subsection{Different Time Scales of Strategy Updates}
\label{different scales of strategy updates}

In our previous simulations, we conducted strategy updates at integer time for individuals. In this simulation, we investigate the impact of different time scales of strategy updates on the evolution of cooperation. Specifically, we explore scenarios where the time interval of strategy update follows fixed values (0.5, 1, 5), exponential and power-law distributions. The evolutionary curves of cooperation fractions over time are illustrated in Figs. \ref{t_f_WS} and \ref{t_f_SL} for WS and SL networks, respectively. The sizes of WS and SL networks are set to 10000, and other parameters are configured as $\lambda_0=\mu_1=0.03, \lambda_1=0.04, \mu_2=0.02, b=1.5$, and $r=0.5$. To ensure that the cooperation frequency in all cases reaches a stationary level, we set the evolution time to 20000. Furthermore, we utilize a logarithmic scale for the $x$-axis to better observe the ascent and descent stages in the evolutionary process.

From Figs. \ref{t_f_WS} and \ref{t_f_SL}, we observe that the cooperation frequency steadily evolves as time progresses in all cases and that the number of cooperators is higher when individuals update their strategies based on exponential and power-law distributions compared to fixed time intervals. Furthermore, both plots demonstrate that the time for the evolution of the cooperation density to plateau increases as the fixed time interval of strategy updates grows. This is because individuals are updating their strategies synchronously in this scenario, and a smaller fixed time interval results in a higher frequency of strategy updates by individuals. The distinction between Figs. \ref{t_f_WS} and \ref{t_f_SL} is that when the time interval of strategy update follows a power-law or exponential distribution, the proportion of cooperators on the WS network eventually reaches 1, while on the SL network, the proportion of cooperators does not evolve to 1 and the ratio of cooperators with exponential distribution is higher than that with the power-law distribution.

\begin{center}
\begin{figure}[htbp]
\centering
\subfigure[WS]{
\includegraphics[scale=0.14]{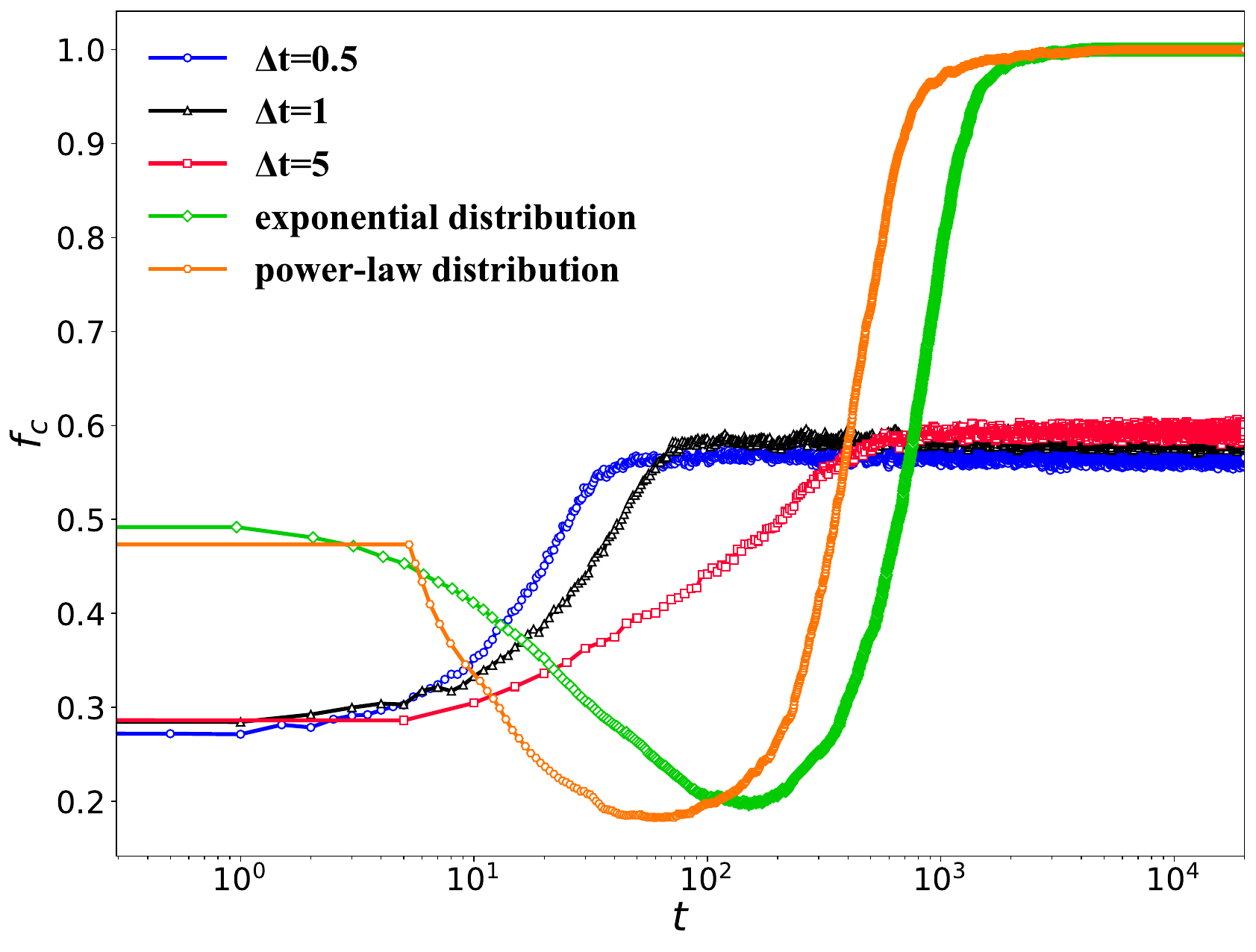}
\label{t_f_WS}}
%\hspace{0.5cm}
\subfigure[SL]{
\includegraphics[scale=0.14]{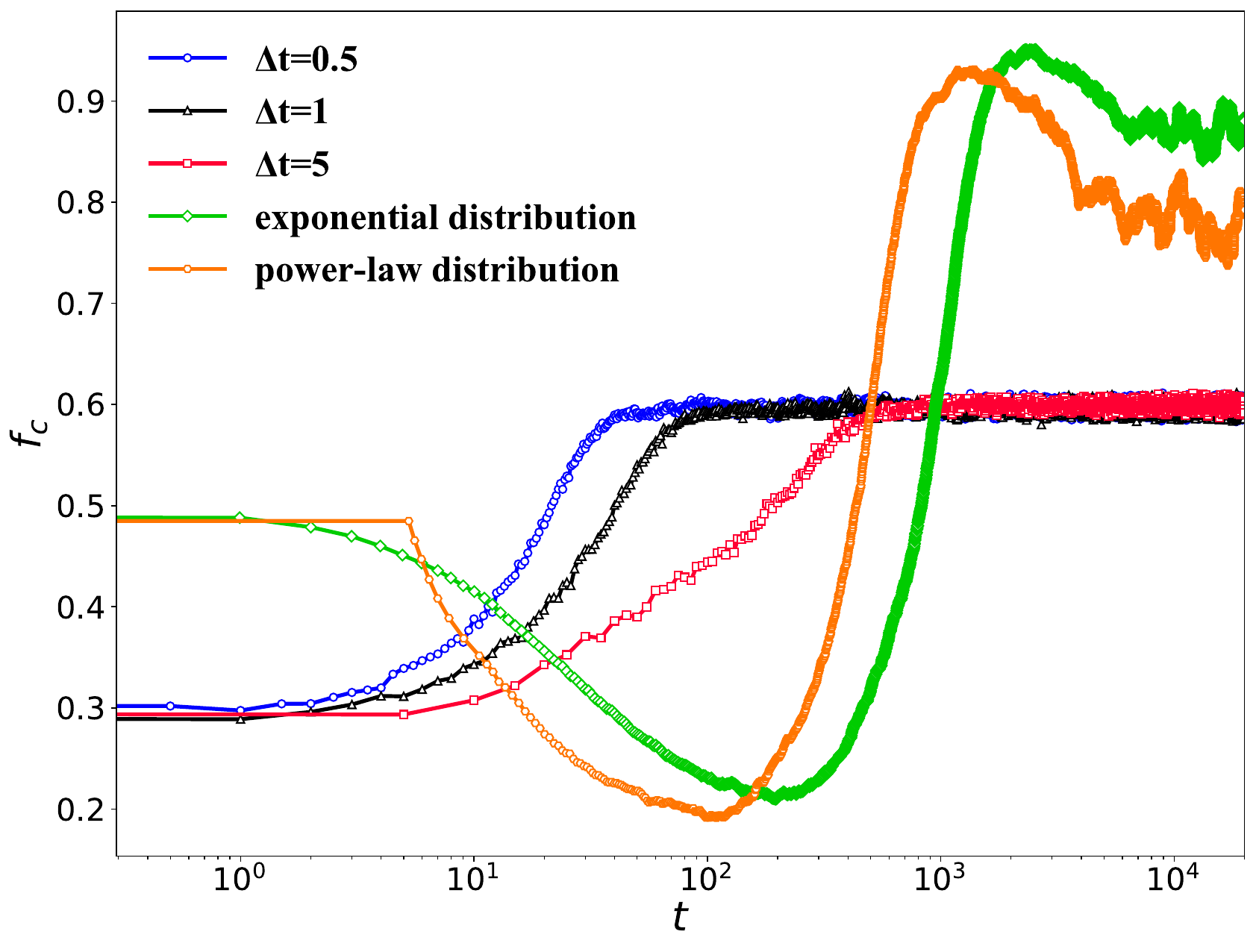}
\label{t_f_SL}}
\caption{\textbf{Evolutionary curves of cooperation frequency under different time scales of strategy updates.} This figure shows the evolution of the cooperation ratio over time on the WS (in panel (a)) and SL (in panel (b)) networks under different time scales of strategy updates, including fixing the time interval for strategy updates to 0.5, 1, 5, and allowing them to obey exponential and power-law distributions, respectively. The results suggest that the cooperation ratio evolves steadily as time progresses in all cases and that the number of cooperators with exponential and power-law distributions is higher than fixed values for the time interval at which individuals update their strategies.}
\label{t_f}
\end{figure}
%\vspace{-1.5\baselineskip}
\end{center}

%\vspace{-2.5\baselineskip}
\subsection{The Influence of Network Scale on the Cooperation Frequency}
\label{effect of network scale}

In this subsection, we evaluate the robustness of the model by examining the influence of network scale on the cooperation frequency of WS and SL networks under different pairs of parameters. We demonstrate the variation curves of the cooperation frequency on the WS and SL networks with respect to the network scale for various parameter combinations in Figs. \ref{N_f_WS} and \ref{N_f_SL}, where the network scales of WS and SL are respectively set as [1000, 21000] and [$30^2$, $210^2$].

\begin{center}
\begin{figure}[htbp]
\centering
\subfigure[WS]{
\includegraphics[scale=0.12]{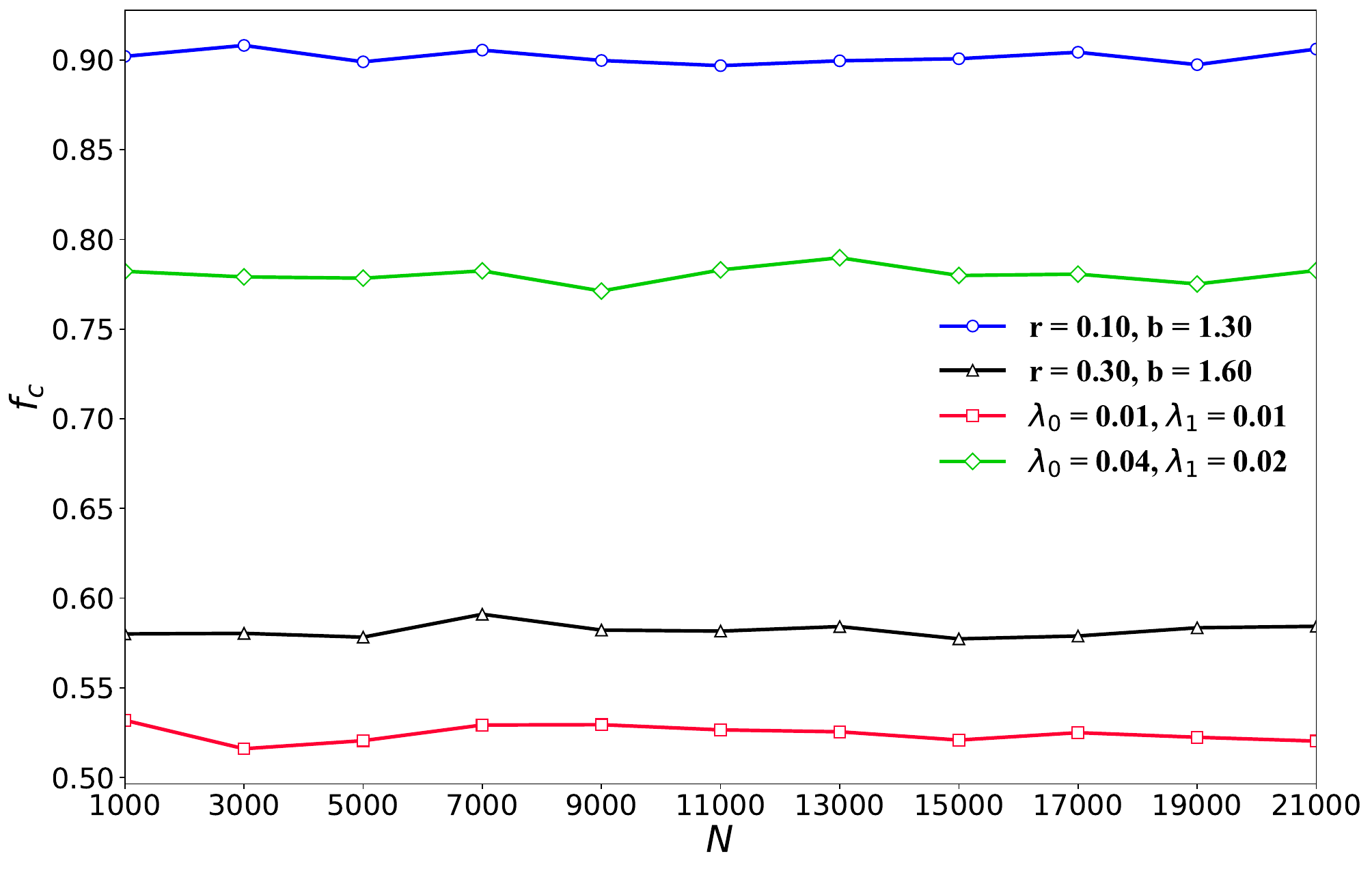}
\label{N_f_WS}}
%\hspace{0.5cm}
\subfigure[SL]{
\includegraphics[scale=0.12]{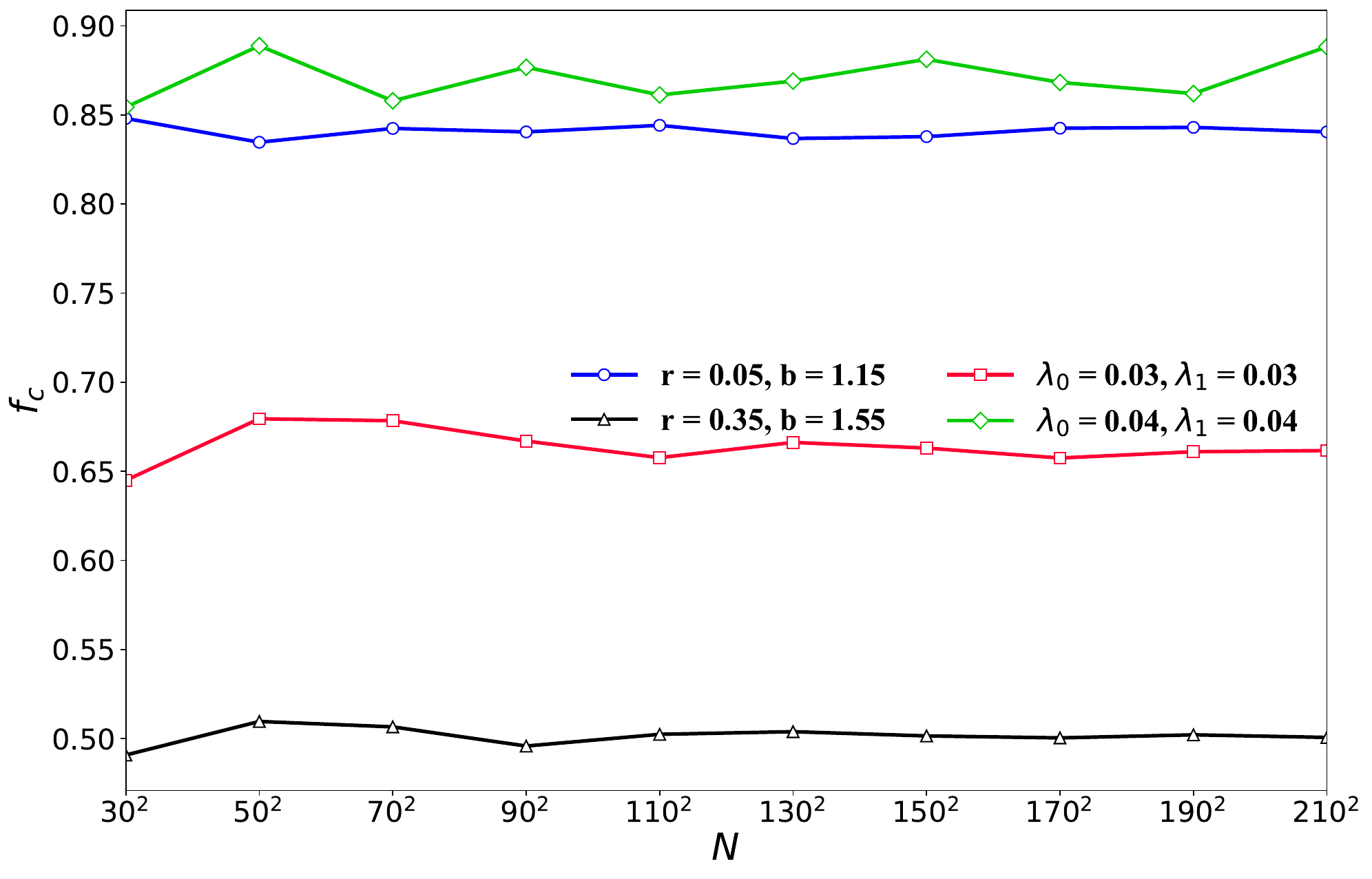}
\label{N_f_SL}}
\caption{\textbf{Plots of cooperation frequency against network scale under different parameters.} We present the cooperation frequency against network scale under different pairs of parameters on the WS (in subplot (a)) and SL (in subplot(b)) networks, respectively, where the $x$-axis is set as network scale, which indicates the number of individuals in the network, while the $y$-axis is set as the cooperation frequency. The ranges of $x$-axis of subplot(a) and subplot(b) are set to [1000, 21000] and [$30^2$, $210^2$], respectively. It can be seen that the network scale has almost no impact on the cooperation frequency on the both WS and SL networks.}
\label{N_f}
\end{figure}
%\vspace{-1.5\baselineskip}
\end{center}

In Figs. \ref{N_f_WS} and \ref{N_f_SL}, we present the evolutionary trends of the cooperator ratio in relation to the network scale for four different parameter pairs on the WS and SL networks respectively. Both WS and SL networks exhibit minimal fluctuation in the evolutionary curves of the cooperator ratio. To quantitatively measure the fluctuation of the cooperation ratio curve, we calculate the variance of the cooperation ratio for the four cases. On the WS network, the variances for the four cases, from top to bottom, are $[1.28, 2.05, 1.33, 2.09]\times 10^{-5}$. On the SL network, the variances for the four cases, from top to bottom, are $[13.88, 1.37, 9.20, 2.47]\times 10^{-5}$. These results indicate that the network scale almost does not affect the cooperative behavior when the network size is relatively large ($N>1000$), indirectly suggesting the robustness of our model. Additionally, we observe that the cooperation frequency represented by blue circles with parameters $r=0.10,b=1.30$ is higher than that represented by black triangles with parameters $r=0.30,b=1.60$. Similarly, the cooperation frequency marked by red squares with parameters $\lambda_0=0.01,\lambda_1=0.01$ is lower than that marked by green diamonds with parameters $\lambda_0=0.04,\lambda_1=0.02$. This implies that increasing the payoff parameters $r$ and $b$ hinders cooperation while enhancing the transition rates $\lambda_0$ and $\lambda_1$ promotes cooperation, which aligns with our previous analysis. This phenomenon and conclusion can also be observed in the SL network, as depicted in Fig. \ref{N_f_SL}.

%\vspace{-0.5\baselineskip}
\section{Conclusion and outlook}    \label{part IV}

Research on network evolutionary games has been providing a framework for understanding the emergence of the cooperative behavior of groups. In this paper, we introduce the game transition to the game between network individuals based on Markov processes, which is substantially different from previous studies. Each individual will transform the game state $G_i$ into the game state $G_{i+1}$ at the rate of $\lambda_i$, whereas the game state $G_{i+1}$ transitions to $G_i$ at the rate of $\mu_{i+1}$, where the duration of each individual staying in a certain game state is subject to an exponential distribution. By giving some definitions, we provide two lemmas and two theorems with their proofs, which illustrate the probability distribution and the expected number of individuals in each game state when they reach stationary. Additionally, the reputation mechanism is introduced into the model to make individuals more inclined to learn from individuals with high reputations when they update their strategies, which fits with reality in some situations. In the simulations, we consider three game states, including PDG, SDG, and SHG. We primarily investigate the individual number staying in different game states over time and their statistical distributions with different transition rates, whose results are in line with our theoretical analysis. Next, we analyze the effect of the transition rates between different game states on the cooperative behavior of the network and find that the transition rates mainly affect the individual game states and yield a change in the individual payoff matrix, which will afford the individual payoffs to change under the same interaction and therefore further lead to a change in the strategies of the individuals. Then, we explore the effect of the payoff parameters $r$ and $b$ on the cooperation ratio in the situations when only taking the game transition into account and taking both game transition and reputation mechanism into account. The results suggest that an increase in either $r$ or $b$ will inhibit the emergence of cooperative behavior and the facilitative effect of cooperation considering both aspects is stronger than that of only considering the game transition. Subsequently, we examine the evolution of the cooperation ratio under various time scales of strategy updates and observe that the cooperation ratio is higher when individuals update their strategies based on exponential and power-law distributions compared to fixed time intervals. Besides, the influence of different network scales on cooperation is also investigated and it is found that the network scale has minimal effect on the number of cooperators, thereby confirming the robustness of our model.

However, there are some other situations that deserve further consideration and study. For example, in this study, the duration of each individual staying in a certain game state follows an exponential distribution, while other probability distributions, such as the logarithmic normal distribution, power-law distribution, and uniform distribution, may yield different theoretical results and new findings. In addition, different strategy-updating rules, such as the Moran process, best-take-over, and replicator dynamics, can also be considered in the update process of the strategy and not only the Fermi function. It is also worth stating that we have investigated the transition of the game between individuals, while the transition of individual interaction is also worth considering. All these issues need further work, which will be the goal of our next stage of research.

%\appendices

%\section{}

\begin{comment}
\section*{Acknowledgment}

This work was supported in part by the National Nature Science Foundation of China (NSFC) under Grant No. 62206230 and No. 12271083, in part by the Humanities and Social Science Fund of Ministry of Education of the People's Republic of China under Grant No. 21YJCZH028, and in part by the Natural Science Foundation of Sichuan Province under Grant No. 2022NSFSC0501.
\end{comment}

\ifCLASSOPTIONcaptionsoff
  \newpage
\fi



%\bibliographystyle{IEEEtran}

\begin{thebibliography}{50}
\bibitem{01}Sigmund K. The calculus of selfishness//The Calculus of Selfishness. Princeton University Press, 2010.
\bibitem{02}Ohtsuki H, Hauert C, Lieberman E, Nowak M A. A simple rule for the evolution of cooperation on graphs and social networks. Nature, 2006, 441(7092): 502-505.
\bibitem{03}Darwin C. The Works of Charles Darwin, Volume 27: The Power of Movement in Plants. NYU Press, 2010.
\bibitem{04}Deng C, Wang L, Rong Z, Wang X. Cooperation emergence in group population with unequal competitions. EPL (Europhysics Letters), 2020, 131(2): 28001.
\bibitem{05}Watts D J, Strogatz S H. Collective dynamics of ¡®small-world¡¯ networks. Nature, 1998, 393(6684): 440-442.
\bibitem{06}Barab$\acute{a}$si A L, Albert R. Emergence of scaling in random networks. Science, 1999, 286(5439): 509-512.
\bibitem{07}Chen D, Liu R, Hu Q, Steven X. Interaction-Aware Graph Neural Networks for Fault Diagnosis of Complex Industrial Processes. IEEE Transactions on Neural Networks and Learning Systems, 2021.
\bibitem{08}Feng M, Li Y, Chen F, Kruths J. Heritable deleting strategies for birth and death evolving networks from a queueing system perspective. IEEE Transactions on Systems, Man, and Cybernetics: Systems, 2022.
\bibitem{09}Li Y, Zeng Z, Feng M, Kurths J. Protection Degree and Migration in the Stochastic SIRS Model: A Queueing System Perspective. IEEE Transactions on Circuits and Systems I: Regular Papers, 2021, 69(2): 771-783.
\bibitem{10}Feng M, Qu H, Yi Z, Xie X, Kurths J. Evolving scale-free networks by Poisson process: Modeling and degree distribution. IEEE Transactions on Cybernetics, 2015, 46(5): 1144-1155.
\bibitem{11}Chiong R, Kirley M. Effects of iterated interactions in multiplayer spatial evolutionary games. IEEE Transactions on Evolutionary Computation, 2012, 16(4): 537-555.
\bibitem{12}Yang Y, Li X. Towards a snowdrift game optimization to vertex cover of networks. IEEE Transactions on Cybernetics, 2013, 43(3): 948-956.
\bibitem{13}Li J, Zhang C, Sun Q, Chen Z, Zhang J. Changing the intensity of interaction based on individual behavior in the iterated prisoner's dilemma game. IEEE Transactions on Evolutionary Computation, 2016, 21(4): 506-517.
\bibitem{14}Zeng Z, Li Y, Feng M. The spatial inheritance enhances cooperation in weak prisoner's dilemmas with agents' exponential lifespan. Physica A: Statistical Mechanics and its Applications, 2022, 593: 126968.
\bibitem{15}Capraro V, Rodriguez-Lara I, Ruiz-Martos M J. Preferences for efficiency, rather than preferences for morality, drive cooperation in the one-shot Stag-Hunt Game. Journal of Behavioral and Experimental Economics, 2020, 86: 101535.
\bibitem{16}Belloc M, Bilancini E, Boncinelli L, D'Alessandro S. Intuition and deliberation in the stag hunt game. Scientific Reports, 2019, 9(1): 1-7.
\bibitem{17}Nowak M A. Five rules for the evolution of cooperation. Science, 2006, 314(5805): 1560-1563.
\bibitem{18}Szolnoki A, Perc M. Leaders should not be conformists in evolutionary social dilemmas. Scientific Reports, 2016, 6(1): 1-8.
\bibitem{19}Pi B, Zeng Z, Feng M, Kurths J. Evolutionary multigame with conformists and profiteers based on dynamic complex networks. Chaos: An Interdisciplinary Journal of Nonlinear Science, 2022, 32(2): 023117.
\bibitem{20}Wang W X, Ren J, Chen G, Wang B H. Memory-based snowdrift game on networks. Physical Review E, 2006, 74(5): 056113.
\bibitem{21}Deng Y, Zhang J. Memory-based prisoner's dilemma game with history optimal strategy learning promotes cooperation on interdependent networks. Applied Mathematics and Computation, 2021, 390: 125675.
\bibitem{22}Fan R, Wang Y, Lin J. Study on Multi-Agent Evolutionary Game of Emergency Management of Public Health Emergencies Based on Dynamic Rewards and Punishments. International Journal of Environmental Research and Public Health, 2021, 18(16): 8278.
\bibitem{23}Andreoni J, Harbaugh W, Vesterlund L. The carrot or the stick: Rewards, punishments, and cooperation. American Economic Review, 2003, 93(3): 893-902.
\bibitem{24}Su Q, McAvoy A, Wang L, Nowak M A. Evolutionary dynamics with game transitions. Proceedings of the National Academy of Sciences, 2019, 116(51): 25398-25404.
\bibitem{25}Hilbe C, $\breve{S}$imsa $\breve{S}$, Chatterjee K, Nowak M A. Evolution of cooperation in stochastic games. Nature, 2018, 559(7713): 246-249.
\bibitem{26}Fu F, Hauert C, Nowak M A, Wang L. Reputation-based partner choice promotes cooperation in social networks. Physical Review E, 2008, 78(2): 026117.
\bibitem{27}Hu Z, Li X, Wang J, Xia C, Wang Z, Perc M. Adaptive Reputation Promotes Trust in Social Networks. IEEE Transactions on Network Science and Engineering, 2021, 8(4): 3087-3098.
\bibitem{28}Luo M, Fan R, Zhang Y, Zhu C. Environmental governance cooperative behavior among enterprises with reputation effect based on complex networks evolutionary game model. International Journal of Environmental Research and Public Health, 2020, 17(5): 1535.
\bibitem{29}Chakrabarti P, Satpathy B, Shankar H A, et al. Statistical Analysis Of Strategic Market Management Based On Neuro-Fuzzy Model Of Human Nature, Poisson Process And Renewal Theory. Ilkogretim Online, 2020, 19(4): 7146-7159.
\bibitem{30}Miritello G, Moro E, Lara R. Dynamical strength of social ties in information spreading. Physical Review E, 2011, 83(4): 045102.
\bibitem{31}Liang X, Zheng X, Lv W, Zhu T, Xu K. The scaling of human mobility by taxis is exponential. Physica A: Statistical Mechanics and its Applications, 2012, 391(5): 2135-2144.
\bibitem{32}Arenas A, D{\'\i}az-Guilera A, Kurths J, Moreno Y, Zhou C. Synchronization in complex networks. Physics reports, 2008, 469(3): 93-153.
\bibitem{33}Kitsak M, Gallos L K, Havlin S, et al. Identification of influential spreaders in complex networks. Nature Physics, 2010, 6(11): 888-893.
\bibitem{34}Vasconcelos V V, Santos F P, Santos F C, Pacheco J M. Stochastic dynamics through hierarchically embedded Markov chains. Physical Review Letters, 2017, 118(5): 058301.
\bibitem{35}Lacasa L, Mari{\~n}o I P, Miguez J, et al. Multiplex decomposition of non-markovian dynamics and the hidden layer reconstruction problem. Physical Review X, 2018, 8(3): 031038.
\end{thebibliography}
%\begin{footnotesize}
%\begin{spacing}{0.9}

\bibliographystyle{ieeetr}
%\end{spacing}
%\end{footnotesize}


%\vspace{-2.5\baselineskip}
\begin{IEEEbiography}[{\includegraphics[width=1in,height=1.25in,clip,keepaspectratio]{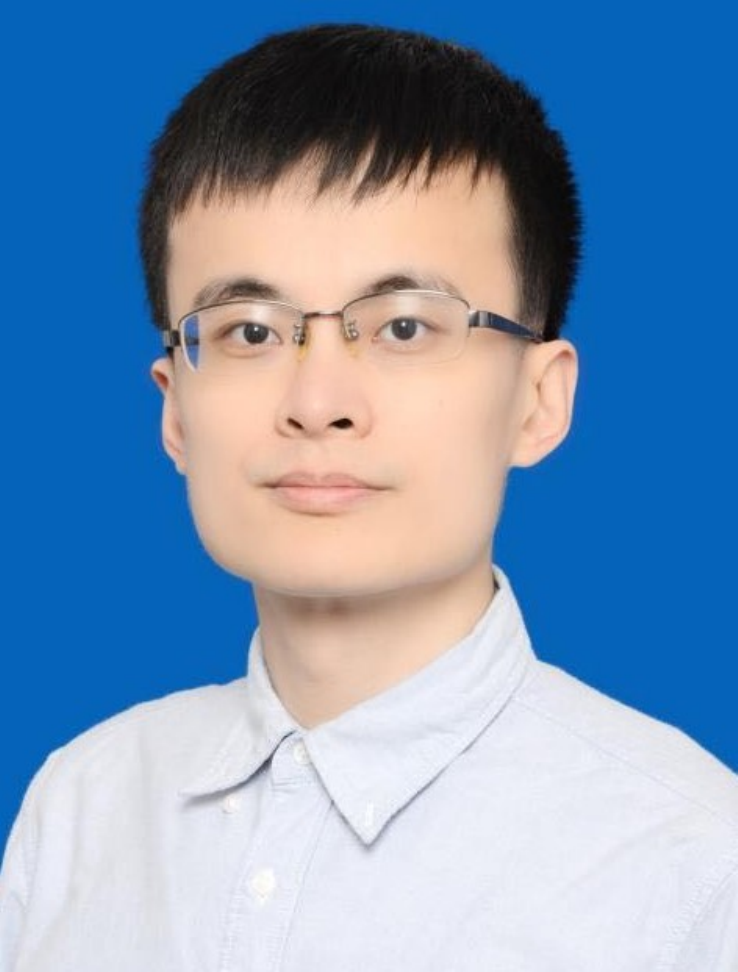}}]{Minyu Feng} (Member IEEE) received the B.S. degree in mathematics from the University of Electronic Science and Technology of China in 2010; the Ph.D. degree in computer science from the University of Electronic Science and Technology of China in 2018. From 2016 to 2017, he was a visiting scholar with the Potsdam Institute for Climate Impact Research, Germany, and Humboldt University, Berlin, Germany. Since 2019, he has been an associate professor in the College of Artificial Intelligence, Southwest University, Chongqing, China. He is a Senior Member of CCF, an academic member of IEEE and CAA. He currently serves as an editorial board member of the International Journal of Mathematics for Industry, also a guest editor in Entropy and Frontiers in Physics. His research interests include complex systems, stochastic processes, evolutionary games, and social computing.
\end{IEEEbiography}

\vspace{-2.3\baselineskip}
\begin{IEEEbiography}[{\includegraphics[width=1in,height=1.25in,clip,keepaspectratio]{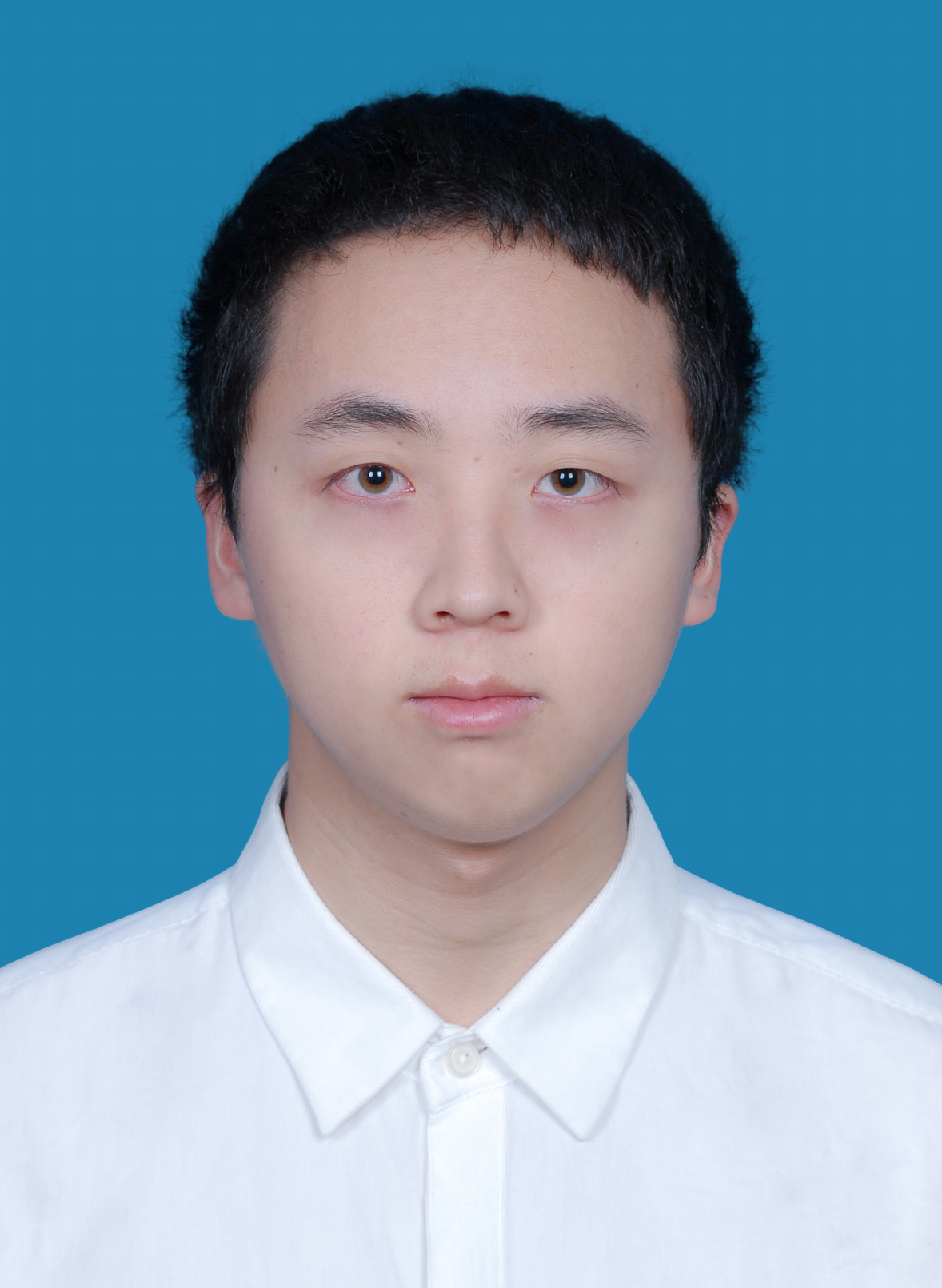}}]{Bin Pi} received the B.E. degree in data science and big data technology from the College of Artificial Intelligence, Southwest University, Chongqing, China. He is currently pursuing the M.S. degree with the School of Mathematical Sciences, University of Electronic Science and Technology of China, Chengdu, China. His current research interests include complex networks, evolutionary games, stochastic processes, and nonlinear science.
\end{IEEEbiography}

\vspace{-2.3\baselineskip}
\begin{IEEEbiography}[{\includegraphics[width=1in,height=1.25in,clip,keepaspectratio]{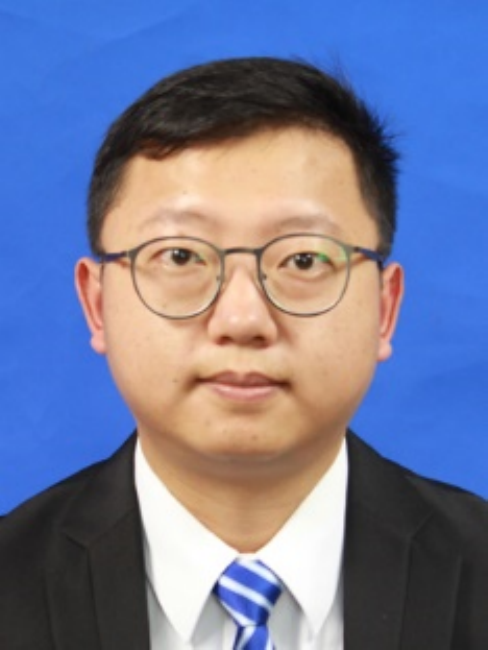}}]{Liang-Jian Deng} (Senior Member, IEEE) received the B.S. and Ph.D. degrees in applied mathematics from the School of Mathematical Sciences, University of Electronic Science and Technology of China (UESTC), Chengdu, China, in 2010 and 2016, respectively. From 2013 to 2014, he was a joint-training Ph.D. Student with Case Western Reserve University, Cleveland, OH, USA. In 2017, he was a Post-Doctoral Researcher with Hong Kong Baptist University (HKBU), Hong Kong. In addition, he has stayed at the Isaac Newton Institute for Mathematical Sciences, University of Cambridge, Cambridge, U.K., and HKBU, for short visits. He is currently a Professor with the School of Mathematical Sciences, UESTC. His research interests include the use of optimization modeling, deep learning, and numerical PDEs, to address several tasks in image processing and computer vision, e.g., resolution enhancement and restoration. Please visit his homepage for more info.: https://liangjiandeng.github.io/.
\end{IEEEbiography}

\vspace{-2.3\baselineskip}
\begin{IEEEbiography}[{\includegraphics[width=1in,height=1.25in,clip,keepaspectratio]{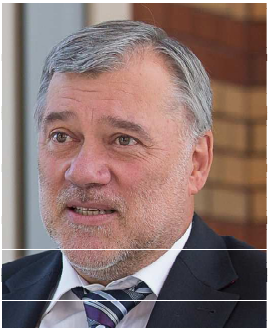}}]{J\"{u}rgen Kurths} received the B.S. degree in mathematics from the University of Rostock, Rostock, Germany, the Ph.D. degree from the Academy of Sciences, German Democratic Republic, Berlin, Germany, in 1983, the Honorary degree from N.I.Lobachevsky State University, Nizhny Novgorod, Russia in 2008, and the Honorary degree from Saratow State University, Saratov, Russia, in 2012.

From 1994 to 2008, he was a Full Professor with the University of Potsdam, Potsdam, Germany. Since 2008, he has been a Professor of nonlinear dynamics with the Humboldt University of Berlin, Berlin, Germany, and the Chair of the Research Domain Complexity Science with the Potsdam Institute for Climate Impact Research, Potsdam, Germany. He has authored more than 700 papers, which are cited more than 60000 times (H-index: 111). His main research interests include synchronization, complex networks, time series analysis, and their applications.

Dr. Kurths was the recipient of the Alexander von Humboldt Research Award from India, in 2005, and from Poland in 2021, the Richardson Medal of the European Geophysical Union in 2013, and the Eight Honorary Doctorates. He is a Highly Cited Researcher in Engineering. He is a member of the Academia 1024 Europaea. He is an Editor-in-Chief of CHAOS and on the Editorial Boards of more than ten journals. He is a Fellow of the American Physical Society, the Royal Society of 1023 Edinburgh, and the Network Science Society.
\end{IEEEbiography}

%\end{IEEEbiography}


%\newpage



%\vfill


\end{document}